\documentclass[journal]{IEEEtran}

\usepackage{cite}
\usepackage{amsmath,amssymb,amsfonts}
\usepackage{graphicx,color}
\usepackage{textcomp}
\usepackage{url}
\usepackage{algorithm}
\usepackage{algpseudocode}
\usepackage{subfig}  

\DeclareMathOperator{\tr}{tr}
\DeclareMathOperator{\diag}{diag}
\DeclareMathOperator{\Reop}{Re}

\newcommand{\CN}{\mathcal{CN}}

\newtheorem{assumption}{Assumption}
\newtheorem{remark}{Remark}
\newtheorem{proposition}{Proposition}
\newtheorem{definition}{Definition}

\begin{document}

\title{Wideband Near-Field Channel Estimation Under Hybrid Compression:
Cross-Subcarrier KL Covariance Fitting With OFDM Fresnel Model}

\author{R{\i}fat Volkan \c{S}enyuva%
\thanks{The author is with the Department of
Electrical-Electronics Engineering, Maltepe University,
Istanbul 34857, Turkey
(e-mail: rifatvolkansenyuva@maltepe.edu.tr).}%
\thanks{The author declares no funding for this work.}%
\thanks{This work has been submitted to the IEEE for possible publication.
Copyright may be transferred without notice, after which this version may
no longer be accessible.}}

\markboth{Wideband Near-Field Channel Estimation Under Hybrid Compression}%
{\c{S}enyuva}

\maketitle

\begin{abstract}
We consider wideband channel estimation for extremely large-scale
multiple-input multiple-output (XL-MIMO) arrays under hybrid
analog-digital compression, in which a uniform linear array (ULA) is
observed through far fewer radio-frequency (RF) chains than antennas.
At a carrier frequency of $28$~GHz with bandwidths reaching several
hundred MHz, the standard narrowband polar-domain channel model fails:
the near-field Fresnel curvature becomes subcarrier-dependent, and the
compressed observation destroys the per-subcarrier spatial covariance
structure that narrowband methods exploit.
We propose the Wideband Cross-subcarrier Kullback--Leibler (WB-CL-KL)
estimator, which jointly estimates angle and range directly from the
compressed sample covariance, without full-array reconstruction, by
fitting a structured Fresnel covariance model across orthogonal
frequency-division multiplexing (OFDM) subcarriers via a
cross-subcarrier Kullback--Leibler (KL) divergence criterion.
We also derive the wideband compressed-domain Cram\'{e}r--Rao bound
(CRB)---the performance lower bound for this hybrid architecture---from
the Slepian--Bangs formula, and decompose its gain over the narrowband
bound into a data-diversity component of $+27.093$~dB and a
geometric-diversity component of $+0.701$~dB, totalling $+27.793$~dB at
$B = 400$~MHz (Propositions~1 and~2).
In the single-path line-of-sight regime, WB-CL-KL attains a range
root-mean-square error of $19.8$~mm against a $19.9$~mm bound at
signal-to-noise ratio (SNR)~$= 10$~dB, a ratio of $0.996$.
Under the 3GPP Urban Micro (UMi) path-loss and shadow-fading
SNR distribution, it achieves a bound ratio of
$0.959$ at the median deployment SNR of $9.6$~dB, indicating near-CRB
operation at the representative deployment point, where the
compressed-domain bound is evaluated at the scene-median geometry.
\end{abstract}

\begin{IEEEkeywords}
channel estimation, Cram\'{e}r--Rao bound, hybrid analog-digital architecture,
near-field communications, OFDM, wideband XL-MIMO
\end{IEEEkeywords}

\section{Introduction}
\label{sec:intro}
\IEEEPARstart{W}{e} consider the problem of wideband channel estimation for
near-field extremely large aperture MIMO (XL-MIMO) systems under hybrid
analog-digital compression.
Extremely large antenna arrays operating in sub-THz and upper millimeter-wave
(mmWave) bands push the communication range into the Fresnel regime, where
spherical wavefronts couple the angle and range of each scatterer and the
conventional plane-wave approximation fails~\cite{liu2023nftutorial,ye2024elaa}.
Wideband orthogonal frequency-division multiplexing (OFDM) signaling with
bandwidths reaching several hundred MHz at a carrier frequency $f_c = 28$~GHz
introduces a frequency-dependent Fresnel curvature: the quadratic phase
curvature across the antenna aperture varies with subcarrier index, invalidating
narrowband polar-domain sparsity models that underpin compressed sensing
approaches designed for a single frequency~\cite{parvini2025wbxlmimo}.
Hybrid analog-digital architectures, in which $N_\mathrm{RF} \ll M$ radio
frequency (RF) chains compress the $M$-antenna observation through a combining
matrix $\mathbf{W}$, introduce a further information loss that is precisely
quantified by the compressed-domain Cram\'{e}r--Rao bound (CRB).
The triple intersection of wideband OFDM, near-field Fresnel geometry, and hybrid
compression defines the estimation problem that this paper addresses, and no
unified treatment of all three constraints appears in the existing literature.

The polar-domain simultaneous orthogonal matching pursuit (SOMP)
estimator of~\cite{cui2022polar}, the
dictionary learning orthogonal matching pursuit (DL-OMP) of~\cite{zhang2024dlomp}, and the bilinear
pattern detection (BPD) algorithm of~\cite{cuidai2023sciencechina} capture
near-field spherical-wavefront geometry but assume narrowband signaling
and full-array ($N_\mathrm{RF} = M$) access, leaving wideband curvature and
compression gaps unaddressed.
Recent wideband near-field estimators address the frequency-dependent
curvature through structured sparse recovery---block- and
cross-subcarrier-sparse Bayesian learning~\cite{chen2026blocksparse},
parametric-symmetry decoupling under beam squint~\cite{lu2026distributed},
and array-perturbation beam-split estimation with Cram\'{e}r--Rao
analysis~\cite{elbir2023beamsplit}---but operate on the array-domain or
per-subcarrier sparse signal rather than the compressed sample covariance,
and none fits a shared cross-subcarrier covariance model under hybrid
compression.
The dual-wideband XL-MIMO estimator of Tang \emph{et al.}~\cite{tang2025dualwb}
jointly models beam-split and spatial non-stationarity through a
message-passing recovery on the full array, but neither operates on a
compressed sample covariance nor derives a compressed-domain bound.
Delay-domain and angle-domain sparse recovery methods for wideband OFDM
handle frequency-selective channels but adopt far-field steering vectors and
therefore cannot estimate range.
The wideband near-field beamforming design of Wang \emph{et al.}~\cite{wang2026icassp}
employs a time-delay-and-sum hybrid architecture under an OFDM Fresnel signal
model closely related to that of this paper, but assumes the channel is known
and solves a beamformer design problem; no channel estimation or CRB analysis
is provided.
The narrowband predecessor of the present work~\cite{senyuva2026clkl} introduces
covariance-domain KL fitting for near-field hybrid-compressed channel estimation,
but is limited to a single carrier frequency and does not derive a
compressed-domain CRB.
The wideband near-field CRBs of Wei \emph{et al.}~\cite{wei2025wbcrb} and Wang
\emph{et al.}~\cite{wang2025twc} quantify fundamental limits for full-array receivers;
neither addresses the information loss due to hybrid compression or provides an
estimator achieving those bounds; the near-field hybrid CRB
of~\cite{thallapalli2026crlb} incorporates time-delay hybrid architectures
under a Fresnel approximation but considers a uniform circular array in a
high signal-to-noise ratio (SNR) regime and derives no estimator.
The Fisher information framework for near-field localization established
in~\cite{wymeersch2020icc} provides the Fisher information matrix (FIM)
decomposition foundations that this
paper extends to the wideband compressed-domain regime.
The narrowband far-field counterpart---covariance-guided DFT beam selection for
beamspace ESPRIT in hybrid arrays---is developed in~\cite{senyuva2026sensors};
the present paper extends the covariance-fitting principle to the near-field
Fresnel regime under wideband OFDM signaling and hybrid compression.
The KL covariance-fitting lineage originates in the maximum-likelihood structured
covariance framework of~\cite{pote2023tsp,stoica1990perf}; neither reference
addresses spherical-wave geometry or hybrid combining.
Standard wideband multipath channel generation follows the 3GPP TR~38.901
cluster-delay-line model~\cite{tr38901_v16}; an accessible MATLAB tutorial
implementation is provided in~\cite{riviello2022}.
No existing work addresses simultaneous wideband OFDM signaling, near-field
Fresnel geometry, hybrid analog-digital compression, and compressed-domain CRB
analysis in a unified framework.
Table~\ref{tab:priorart} contrasts the proposed method with the most closely
related prior art across these axes.

\begin{table*}[t]
\centering
\caption{Positioning of the proposed WB-CL-KL estimator against the most closely
related prior art. A check mark ($\checkmark$) indicates the property is addressed.
WB: wideband OFDM; NF: near-field Fresnel geometry; Hybrid: hybrid analog-digital
compression ($N_\mathrm{RF}\ll M$); Est.: provides a channel estimator; C-CRB:
compressed-domain CRB; Cov.: covariance-domain fitting (vs.\ sparse recovery).
The CRBs of~\cite{wei2025wbcrb,wang2025twc} are derived for full-array receivers.}
\label{tab:priorart}
\small
\begin{tabular}{l c c c c c c}
\hline\hline
Method & WB & NF & Hybrid & Est. & C-CRB & Cov. \\
\hline
Cui \& Dai~\cite{cui2022polar}                 &            & \checkmark &            & \checkmark &            &            \\
Cui \& Dai~\cite{cuidai2023sciencechina}       & \checkmark & \checkmark &            & \checkmark &            &            \\
Chen \emph{et al.}~\cite{chen2026blocksparse}  & \checkmark & \checkmark &            & \checkmark &            &            \\
Zhang \emph{et al.}~\cite{zhang2024dlomp}      & \checkmark & \checkmark &            & \checkmark &            &            \\
Tang \emph{et al.}~\cite{tang2025dualwb}       & \checkmark & \checkmark &            & \checkmark &            &            \\
Elbir \emph{et al.}~\cite{elbir2023beamsplit}  & \checkmark & \checkmark & \checkmark & \checkmark & \checkmark &            \\
Lu \emph{et al.}~\cite{lu2026distributed}      & \checkmark & \checkmark & \checkmark & \checkmark &            &            \\
Thallapalli \emph{et al.}~\cite{thallapalli2026crlb} & \checkmark & \checkmark & \checkmark &    & \checkmark &            \\
Wei/Wang~\cite{wei2025wbcrb,wang2025twc}       & \checkmark & \checkmark &            &            &            &            \\
Pote \& Rao~\cite{pote2023tsp}                 &            &            &            & \checkmark &            & \checkmark \\
\hline
\textbf{Proposed (WB-CL-KL)}                   & \checkmark & \checkmark & \checkmark & \checkmark & \checkmark & \checkmark \\
\hline\hline
\end{tabular}
\end{table*}

The contributions of this paper are as follows.
First, we propose a wideband cross-subcarrier KL covariance fitting estimator
(WB-CL-KL) that jointly estimates angle and range under the near-field OFDM
Fresnel model with hybrid analog-digital compression.
The estimator exploits cross-subcarrier covariance structure to resolve the
frequency-dependent curvature ambiguity that defeats narrowband polar-domain
methods, and employs SNR-adaptive diagonal loading to stabilize the KL objective
across operating conditions.
Second, we derive the wideband compressed-domain CRB under hybrid combining and
decompose the total CRB gain into a data diversity component, scaling as
$10\log_{10}(K_s)$~dB with the number of subcarriers $K_s$, and a geometric
diversity component arising from frequency-dependent Fresnel curvature.
At $B = 400$~MHz, the total CRB gain relative to the narrowband limit reaches
$+27.793$~dB, with data diversity contributing $+27.093$~dB and geometric
diversity contributing $+0.701$~dB.
Third, we validate WB-CL-KL against the compressed-domain CRB through Monte Carlo
simulation, demonstrating root-mean-square error RMSE$_r = 0.0198$~m versus CRB$_r = 0.0199$~m
(ratio $0.996$) at SNR$\,=10$~dB and normalized mean-squared error NMSE$_r = -43.16$~dB at $B = 400$~MHz,
demonstrating near-CRB operation in the strong near-field regime, where
the compressed-domain bound is evaluated at the scene-median geometry.
A companion conference paper~\cite{senyuva2026globecom} presented the CRB
derivation and diversity decomposition; the present paper extends that work with
the WB-CL-KL estimator, full Monte Carlo validation, and estimator-CRB gap
analysis.

The remainder of this paper is organized as follows.
Section~\ref{sec:model} presents the signal model and system model.
Section~\ref{sec:method} develops the WB-CL-KL estimator.
Section~\ref{sec:crb} derives the wideband compressed-domain CRB.
Section~\ref{sec:results_est} evaluates estimator performance via Monte Carlo
simulation, and Section~\ref{sec:results_crb} presents the CRB and diversity
decomposition results.
Section~\ref{sec:conclusion} concludes.
Throughout this paper, $\mathbf{A}$ denotes a matrix, $\mathbf{a}$ a column
vector, $(\cdot)^T$, $(\cdot)^H$, $(\cdot)^{-1}$, and $(\cdot)^\dagger$ denote
transpose, Hermitian transpose, inverse, and pseudoinverse, respectively.
The expectation operator is written $\mathbb{E}[\cdot]$, and $\|\cdot\|_F$ denotes the Frobenius
norm.

\section{Signal Model and System Model}
\label{sec:model}

\subsection{Near-Field XL-ULA Geometry}
\label{sec:model:geo}

We consider an $M$-element uniform linear array (ULA) with half-wavelength
inter-element spacing $d_\mathrm{ant} = \lambda_c/2$, where
$\lambda_c = c/f_c$ is the carrier wavelength.
Let $\bar{m} = m - (M-1)/2$ denote the centered element index,
$m = 0,\ldots,M-1$.
Under the Fresnel (uniform spherical-wave, USW) approximation, the
$m$th entry of the narrowband near-field steering vector for a scatterer
at angle $\theta$ and range $r$ is
\begin{equation}
  [\mathbf{a}(\theta,r)]_m =
    \exp\!\big(j\omega(\theta)\,\bar{m}
             - j\kappa(\theta,r)\,\bar{m}^2\big),
  \label{eq:a_nb}
\end{equation}
where the linear and quadratic phase coefficients are defined as
\begin{align}
  \omega(\theta)   &= -\frac{2\pi d_\mathrm{ant}}{\lambda_c}\cos\theta,
  \label{eq:omega} \\
  \kappa(\theta,r) &= \phantom{-}\frac{\pi d_\mathrm{ant}^2}{\lambda_c}
                      \frac{\sin^2\!\theta}{r}.
  \label{eq:kappa}
\end{align}
We use $u = 1/r$ as the range variable throughout, since the curvature
coefficient $\kappa = (\pi d_\mathrm{ant}^2/\lambda_c)\sin^2\!\theta \cdot u$
is linear in $u$, which simplifies the gradient expressions in
Section~\ref{sec:method}.
The Fresnel approximation is valid when $r$ lies within the Rayleigh
distance $r_\mathrm{RD} = 2D_\mathrm{ap}^2/\lambda_c$, where
$D_\mathrm{ap} = (M-1)d_\mathrm{ant}$ is the array aperture.

Assumption~\ref{asm:pilot} below requires cross-subcarrier independence
of the path-gain realisations, which corresponds to the standard OFDM
pilot-design contract.
This is consistent with the cross-frequency geometry consistency rules
of 3GPP TR~38.901 Sec.~7.6.5~\cite{tr38901_v16}, which mandate that
cluster geometry is shared across frequency bands while per-cluster
shadowing and small-scale phases are independently regenerated.

\subsection{OFDM Wideband Model}
\label{sec:model:ofdm}

The base station (BS) receives an OFDM waveform with $K$ subcarriers spaced by
$\Delta f$, yielding total bandwidth $B = K\Delta f$.
The $k$th subcarrier frequency is $f_k = f_c + (k - K/2)\Delta f$,
$k = 1,\ldots,K$, and we define the frequency ratio
\begin{equation}
  \alpha_k = \frac{f_k}{f_c}
    = 1 + \frac{(k-K/2)\Delta f}{f_c}.
  \label{eq:alpha_k}
\end{equation}
The ratio $\alpha_k$ scales both the linear and quadratic phase of the
steering vector, coupling angle and range to subcarrier frequency.
For the $\ell$th propagation path, the frequency-scaled Fresnel steering
vector at subcarrier $k$ is
\begin{equation}
  [\mathbf{a}_{\ell,k}]_m =
    \exp\!\big(j\,\alpha_k\,\omega_\ell\,\bar{m}
             - j\,\alpha_k\,\kappa_\ell\,\bar{m}^2\big),
  \label{eq:a_k}
\end{equation}
where $\omega_\ell = \omega(\theta_\ell)$ and
$\kappa_\ell = \kappa(\theta_\ell, r_\ell)$.
As $\alpha_k$ departs from unity, the beam peak and Fresnel curvature
shift jointly, a phenomenon absent in narrowband models.
The wideband near-field OFDM Fresnel signal model in~\eqref{eq:a_k}
is the closest published signal-model twin of the time-delay-and-sum
(TTD) hybrid architecture in~\cite{wang2026icassp}; our
fractional bandwidth $B/f_c \approx 1.4\%$ is substantially narrower
than the ${\sim}10\%$ considered there, so frequency-flat phase-shifter
combining remains a valid approximation.

The per-subcarrier element-space observation at snapshot $n$ is
\begin{equation}
  \mathbf{x}_k(n) = \sum_{\ell=1}^{d}
    s_{\ell,k}(n)\,\mathbf{a}_{\ell,k} + \mathbf{w}_k(n),
  \label{eq:x_k}
\end{equation}
where $s_{\ell,k}(n)$ is the complex path gain at subcarrier $k$,
snapshot $n$, and $\mathbf{w}_k(n) \in \mathbb{C}^M$ is additive noise.

\subsection{Hybrid Analog-Digital Compression}
\label{sec:model:hybrid}

A hybrid analog-digital architecture with $N_\mathrm{RF} \ll M$ RF
chains applies a frequency-flat combining matrix
$\mathbf{W} \in \mathbb{C}^{M \times N_\mathrm{RF}}$ to compress the
$M$-dimensional observation before digital processing.
The compressed observation at subcarrier $k$, snapshot $n$ is
\begin{equation}
  \mathbf{y}_k(n) = \mathbf{W}^H\mathbf{x}_k(n)
    = \sum_{\ell=1}^{d} s_{\ell,k}(n)\,\mathbf{d}_{\ell,k}
      + \mathbf{W}^H\mathbf{w}_k(n),
  \label{eq:y_k}
\end{equation}
where $\mathbf{d}_{\ell,k} = \mathbf{W}^H\mathbf{a}_{\ell,k}
\in \mathbb{C}^{N_\mathrm{RF}}$ is the compressed steering vector.
The combiner $\mathbf{W}$ is frequency-flat across subcarriers (phase-shifter
architecture), so the same compression is applied at every $k$;
frequency-dependent combiners such as TTD architectures are an
open extension discussed in Section~\ref{sec:conclusion}.
The entries of $\mathbf{W}$ satisfy
$|[\mathbf{W}]_{m,j}| = 1/\sqrt{M}$ (constant-modulus constraint);
in simulations $\mathbf{W}$ is drawn uniformly at random from the
set of constant-modulus matrices and its statistics are
characterized in Section~\ref{sec:results_crb}.

\subsection{Compressed-Domain Covariance and Wideband Mismatch}
\label{sec:model:cov}

In this subsection, we derive the compressed-domain covariance model
and quantify the wideband mismatch that motivates the cross-subcarrier
estimator of Section~\ref{sec:method}.

Under Assumption~\ref{asm:pilot} (stated immediately below), the
compressed-domain covariance at subcarrier $k$ is
\begin{equation}
  \mathbf{R}_{y,k}(\boldsymbol{\eta}) =
    \sum_{\ell=1}^{d} p_\ell\,\mathbf{d}_{\ell,k}\mathbf{d}_{\ell,k}^H
    + N_0\,\mathbf{W}^H\mathbf{W},
  \label{eq:Ry_k}
\end{equation}
where $\boldsymbol{\eta} = [\omega_1,\ldots,\omega_d,\,
\kappa_1,\ldots,\kappa_d,\,p_1,\ldots,p_d,\,N_0]^T
\in \mathbb{R}^{3d+1}$
is the shared parameter vector, $p_\ell$ are the path powers, and
$N_0$ is the noise power.
The parameters $\boldsymbol{\eta}$ are shared across all subcarriers;
only the $\alpha_k$-scaling of the steering vector changes with $k$.
The sample covariance is estimated from $N$ snapshots as
\begin{equation}
  \widehat{\mathbf{R}}_{y,k} =
    \frac{1}{N}\sum_{n=1}^{N}\mathbf{y}_k(n)\mathbf{y}_k(n)^H.
  \label{eq:Rhat_k}
\end{equation}

An estimator designed at the center frequency $f_c$ (i.e., using
$\mathbf{R}_{y,k_c}$ for all $k$) suffers a model-mismatch bias that
grows with bandwidth.
The Frobenius-norm relative mismatch
$\|\mathbf{R}_{y,k}-\mathbf{R}_{y,k_c}\|_F / \|\mathbf{R}_{y,k_c}\|_F$
at the edge subcarrier reaches 63.65\% at $B = 100$~MHz, 176.59\% at
$B = 400$~MHz, and 194.18\% at $B = 800$~MHz
($f_c = 28$~GHz, $M = 64$, $r \in [1.06, 4.25]$~m), motivating the
wideband treatment in Sections~\ref{sec:method}--\ref{sec:crb}.
In practice, $K_s \le K$ uniformly-spaced subcarriers suffice for
both CRB evaluation and the WB-CL-KL estimator; the selection
criterion is given in Section~\ref{sec:method}.

%

\begin{assumption}[Stochastic OFDM pilot signal model]
\label{asm:pilot}
For all path indices $\ell \in \{1,\ldots,d\}$, subcarrier indices
$k \in \{1,\ldots,K\}$, and snapshot indices $n \in \{1,\ldots,N\}$:
(A1) the path gains satisfy $s_{\ell,k}(n) \sim \CN(0, p_\ell)$ with
equal powers $p_\ell = 1/d$;
(A2) the collection $\{s_{\ell,k}(n)\}_{\ell,k,n}$ is mutually independent
across all three indices, that is, $s_{\ell,k}(n)$ and $s_{\ell',k'}(n')$
are independent whenever $(\ell,k,n) \neq (\ell',k',n')$;
(A3) the receiver noise satisfies
$\mathbf{w}_k(n) \sim \CN(\mathbf{0}, N_0 \mathbf{I}_M)$, mutually
independent across $(k,n)$ and independent of
$\{s_{\ell,k}(n)\}_{\ell,k,n}$.
The Cram\'{e}r--Rao analysis in Section~\ref{sec:crb} adopts the
stochastic (unconditional) form of this model; the choice over the
conditional alternative is discussed in Remark~\ref{rem:c8}.
\end{assumption}

\begin{remark}[Physical interpretation, and role]
\label{rem:pilot}
Clauses (A1)--(A3) encode the standard OFDM pilot-design contract:
orthogonal pilot symbols probe an uncorrelated multipath channel
across subcarriers and snapshots. Cross-subcarrier independence
in (A2) is the pivotal clause for the wideband Fisher information
decomposition: it is precisely the condition under which the joint
log-likelihood factorises over $k$ and the wideband FIM reduces to
$\mathbf{J}_{\mathrm{WB}} = \sum_{k=1}^{K_s} \mathbf{J}_k$, on which
the $10\log_{10} K_s$ data-diversity term in Proposition~\ref{prop:data_div}
directly depends. In the narrowband specialization $K = 1$, the
cross-$k$ clauses become vacuous and Assumption~\ref{asm:pilot}
reduces to the temporal i.i.d.\ pilot model of~\cite{senyuva2026clkl}. 
Partial violation of (A2) along the $k$ dimension (a correlation envelope 
of effective rank $K_{\mathrm{eff}} \le K_s$) rescales the diversity factor to
$10\log_{10} K_{\mathrm{eff}}$ rather than invalidating the framework,
so the model degrades gracefully. Inter-path independence in (A2)
aligns with line-of-sight and weakly clustered near-field scenarios;
correlated multipath is treated as a generalisation in
Section~\ref{sec:multipath}.
\end{remark}


\begin{remark}[Stochastic versus conditional model]
\label{rem:c8}
Integrating out the path gains under Assumption~\ref{asm:pilot} gives
$\mathbf{y}_k(n)\sim\CN(\mathbf{0},\mathbf{R}_{y,k}(\boldsymbol{\eta}))$,
so the per-subcarrier negative log-likelihood reduces (up to constants)
to the Kullback--Leibler divergence $\log\det\mathbf{R}_{y,k}+
\mathrm{tr}(\mathbf{R}_{y,k}^{-1}\widehat{\mathbf{R}}_{y,k})$; the
wideband objective of Section~\ref{sec:method} is therefore the
stochastic maximum-likelihood
estimator~\cite{stoica1990perf,pote2023tsp}, and joint independence
under~(A1)--(A3) yields the additive Slepian--Bangs decomposition
$\mathbf{J}_{\mathrm{WB}}=\sum_{k=1}^{K_s}\mathbf{J}_k$ that
underlies Proposition~\ref{prop:data_div}. The conditional
alternative of~\cite{stoica1990perf} promotes the $2dK_sN$ gain
realisations to unknowns, replaces this FIM with a data-dependent
block-coupled form that does not factorise over~$k$, and is
asymptotically no tighter for Gaussian gains.
\end{remark}


\begin{remark}[Element-amplitude approximation]
\label{rem:c3}
The phase-only Fresnel/USW steering vector~\eqref{eq:a_k} drops the
per-element amplitude factor $r_\ell/r_{\ell,m}$ that the exact
spherical-wave model carries, where
$r_{\ell,m} = (r_\ell^2 - 2 r_\ell \bar{m} d_\mathrm{ant}\sin\theta_\ell
              + \bar{m}^2 d_\mathrm{ant}^2)^{1/2}$
is the element-to-scatterer distance. Taylor expansion in
$u = \bar{m} d_\mathrm{ant}/r_\ell$ gives
$r_\ell/r_{\ell,m} = 1 + u\sin\theta_\ell + O(u^2)$, with the linear
term integrating to zero across the centred ULA, so the leading
covariance perturbation is of order
$(D_\mathrm{ap}/r_\ell)^2$ with
$D_\mathrm{ap} = (M-1)d_\mathrm{ant}$. By the Slepian--Bangs
sensitivity bound, the relative CRB bias scales as
$O((D_\mathrm{ap}/r)^2)$ and remains below 1~dB throughout the
operating range $r \ge r_\mathrm{lo} \approx 1.06$~m considered in Section~\ref{sec:results_est},
consistent with the convention in prior near-field CRB
analysis~\cite{grosicki2005wlp,wei2025wbcrb,wang2025twc} and with
the narrowband precursor~\cite{senyuva2026clkl}. Including the
amplitude factor would not perturb Assumption~\ref{asm:pilot}, since
$r_\ell/r_{\ell,m}$ is deterministic in the geometry; the omission
is therefore an approximation in the model, not in the statistical
assumptions.
\end{remark}

\section{Wideband Cross-Subcarrier KL Covariance Fitting}
\label{sec:method}

\subsection{Cross-Subcarrier KL Objective}
\label{sec:method:obj}

We extend the narrowband CL-KL covariance fitting
of~\cite{senyuva2026clkl} to the wideband OFDM setting by summing
the per-subcarrier KL divergence over $K_s$ uniformly-spaced
subcarriers.
The maximum identifiable path count is bounded by
$d \le \lfloor(N_\mathrm{RF}-1)/2\rfloor$,
the compressed-array analogue of the Ottersten--Stoica--Roy
condition~\cite{ottersten1998cov}.

The wideband KL objective is
\begin{equation}
  \mathcal{L}(\boldsymbol{\eta}) =
    \sum_{k=1}^{K_s} \left[
      \log\det\mathbf{R}_{y,k}(\boldsymbol{\eta})
      + \mathrm{tr}\!\big(\mathbf{R}_{y,k}^{-1}(\boldsymbol{\eta})
        \widehat{\mathbf{R}}_{y,k}\big)
    \right] + \nu\|\mathbf{p}\|_1,
  \label{eq:WB_KL}
\end{equation}
minimised over $\mathbf{p} \succeq \mathbf{0}$,
$\boldsymbol{\omega} \in [-\pi,\pi]^d$,
$\boldsymbol{\kappa} \in [0, \kappa_{\max}]^d$, $N_0 \ge 0$,
where $\nu > 0$ is the sparsity regularization weight
and $\boldsymbol{\eta} = [\omega_1,\ldots,\omega_d,
\kappa_1,\ldots,\kappa_d, p_1,\ldots,p_d, N_0]^T$
is the shared parameter vector defined in~\eqref{eq:Ry_k}.
The shared parameter structure of $\boldsymbol{\eta}$ means that
a single update to $\omega_\ell$ or $\kappa_\ell$ propagates
simultaneously across all $K_s$ covariance terms.

The per-subcarrier KL residual matrix is
\begin{equation}
  \mathbf{G}_k \triangleq
    \mathbf{R}_{y,k}^{-1}
    - \mathbf{R}_{y,k}^{-1}
      \widehat{\mathbf{R}}_{y,k}
      \mathbf{R}_{y,k}^{-1}.
  \label{eq:Gk}
\end{equation}
The residual matrix $\mathbf{G}_k$ has spectral norm
$\|\mathbf{G}_k\|_2 \propto N_0^{-2}$, which motivates the
frozen-noise strategy in Section~\ref{sec:method:loop}.

Equation~\eqref{eq:WB_KL} is the wideband stochastic maximum-likelihood
criterion~\cite{stoica1990perf,pote2023tsp}, summed over $K_s$ independent
subcarrier observations; the single-subcarrier specialisation $K_s = 1$
recovers the narrowband CL-KL of~\cite{senyuva2026clkl}.
The additive structure $\mathbf{J}_\mathrm{WB} = \sum_{k=1}^{K_s}\mathbf{J}_k$
that justifies this sum is the compressed-domain generalisation of the
Fisher information decomposition established in~\cite{wymeersch2020icc}.

\subsection{Power Gradient and Wideband Curvature Gradient}
\label{sec:method:grad}

The $\alpha_k$-scaling of the wideband steering vector introduces a
subcarrier-frequency chain-rule factor absent in the narrowband case.
We derive the gradients of $\mathcal{L}$ with respect to the path
powers $p_\ell$ and the curvature parameters $\kappa_\ell$.

The power gradient sums over subcarriers the per-path inner product
against the residual $\mathbf{G}_k$:
\begin{equation}
  \frac{\partial \mathcal{L}}{\partial p_\ell}
    = \sum_{k=1}^{K_s}
      \mathbf{d}_{\ell,k}^H \mathbf{G}_k \mathbf{d}_{\ell,k}
      + \nu.
  \label{eq:grad_p_WB}
\end{equation}

The curvature gradient is the key wideband extension relative
to~\cite{senyuva2026clkl}.
In the narrowband model, $\partial\mathbf{d}_i/\partial u_i$ carries
no frequency factor.
In the wideband model, differentiating
$\mathbf{d}_{\ell,k} = \mathbf{W}^H\mathbf{a}_{\ell,k}$ with respect
to $\kappa_\ell$ yields
\begin{equation}
  \frac{\partial \mathbf{d}_{\ell,k}}{\partial \kappa_\ell}
    = -j\,\alpha_k\,\mathbf{W}^H
      \diag(\bar{\mathbf{m}}^{\odot 2})
      \mathbf{a}_{\ell,k},
  \label{eq:dd_dkappa}
\end{equation}
where $\bar{\mathbf{m}} = [-(M-1)/2, \ldots, (M-1)/2]^T$ is the centred
element-index vector and $\bar{\mathbf{m}}^{\odot 2}$ denotes elementwise
squaring.
The wideband curvature gradient is then
\begin{equation}
  \frac{\partial \mathcal{L}}{\partial \kappa_\ell}
    = 2\sum_{k=1}^{K_s} p_\ell\,\Reop\!\left\{
        \left(\frac{\partial \mathbf{d}_{\ell,k}}
                   {\partial \kappa_\ell}\right)^H
        \mathbf{G}_k\,\mathbf{d}_{\ell,k}
      \right\}.
  \label{eq:grad_kappa_WB}
\end{equation}
The $\alpha_k$ factor in~\eqref{eq:dd_dkappa} couples the curvature
sensitivity to subcarrier frequency; at the center subcarrier
$\alpha_{k_c} = 1$, \eqref{eq:grad_kappa_WB} reduces to the narrowband
gradient of~\cite{senyuva2026clkl}.
Equivalently, differentiating with respect to
$u_\ell = 1/r_\ell$ gives the gradient used in
the MATLAB implementation~\cite{senyuva2026globecom}.

\subsection{Power-Only Main Loop with Frozen Noise}
\label{sec:method:loop}

The core algorithmic structure follows the narrowband
predecessor~\cite{senyuva2026clkl}: a \emph{power-only main loop}
with frozen noise estimate, followed by a global joint scan.

\noindent\textbf{Why curvature is removed from the main loop.}~
The curvature gradient~\eqref{eq:grad_kappa_WB} inherits
$\|\mathbf{G}_k\|_2 \propto N_0^{-2}$ from the residual matrix.
At SNR $= 20$~dB the effective gradient magnitude is
$\approx 10^4\times$ that at SNR $= 0$~dB, causing atoms to traverse
the entire valid $\kappa$-range in a single step regardless of curvature
signal quality.
The post-loop matched-filter scan of Section~\ref{sec:method:scan}
replaces the gradient step with a direct maximisation over the full
valid range, which is SNR-invariant by construction.

\noindent\textbf{Frozen noise estimate.}~
$\widehat{N}_0$ is estimated once from the subcarrier-averaged sample
covariance and held fixed throughout the main loop to prevent the
gradient magnitude from diverging at high SNR:
\begin{equation}
  \widehat{N}_0 = \max\!\left(
    \frac{1}{N_\mathrm{RF}-d}
    \sum_{j=1}^{N_\mathrm{RF}-d}
    \lambda_j^{\downarrow}\!\!\left(
      \tfrac{1}{2}\bigl(\widehat{\mathbf{R}}_{y,\mathrm{m}}
      + \widehat{\mathbf{R}}_{y,\mathrm{m}}^H\bigr)
    \right),\;
    \eta_0
  \right),
  \label{eq:N0_frozen}
\end{equation}
where $\eta_0=10^{-12}$ is a numerical floor that prevents
$\widehat{N}_0$ from collapsing to zero at high SNR, $\widehat{\mathbf{R}}_{y,\mathrm{m}} =
(1/K_s)\sum_{k=1}^{K_s}\widehat{\mathbf{R}}_{y,k}$ is the mean
sample covariance across subcarriers and
$\lambda_j^{\downarrow}(\cdot)$ denotes the $j$th smallest eigenvalue.
In the narrowband case $K_s = 1$, \eqref{eq:N0_frozen} reduces
to the single compressed covariance used in~\cite{senyuva2026clkl};
in the wideband case the mean covariance averages out subcarrier-specific
noise fluctuations before the eigenvalue estimator is applied.

\noindent\textbf{Power update.}~
Since $\partial\mathbf{R}_{y,k}/\partial p_\ell =
\mathbf{d}_{\ell,k}\mathbf{d}_{\ell,k}^H$, the power update with
Armijo backtracking ($\alpha_p = 1$, $\beta = 0.5$,
$\sigma = 10^{-4}$) is
\begin{equation}
  p_\ell \leftarrow \max\!\big\{0,\;
    p_\ell - \alpha_p \nabla_{p_\ell}\mathcal{L}\big\}.
  \label{eq:power_update}
\end{equation}
We declare convergence when
$|\mathcal{L}^{(t)} - \mathcal{L}^{(t-1)}| /
|\mathcal{L}^{(t)}| < 5 \times 10^{-4}$ or after
$T_{\max} = 200$ iterations.

\noindent\textbf{SNR-adaptive diagonal loading.}~
To prevent numerical rank deficiency of $\mathbf{R}_{y,k}$ at low SNR,
a diagonal loading term $\varepsilon_\mathrm{reg}
\|\widehat{\mathbf{R}}_{y,k}\|_F / N_\mathrm{RF}$ is added to each
per-subcarrier model covariance before inversion, with
$\varepsilon_\mathrm{reg} = 10^{-3}$ fixed across all operating points.

\subsection{Post-Loop Joint Scan and Multi-Start}
\label{sec:method:scan}
\label{sec:multipath}

After the power-only loop converges, the active set
$\mathcal{S} = \{\ell : p_\ell > 0\}$ (top $d$ by power) undergoes
a global joint scan and a BPD-anchored multi-start strategy.

\noindent\textbf{Post-loop joint scan.}~
For each $\ell \in \mathcal{S}$, we perform four alternating scan passes
(two over $\omega$, two over $\kappa$).
The per-path residual covariance at subcarrier $k$ is
\begin{equation}
  \widetilde{\mathbf{R}}_{\ell,k}
    = \widehat{\mathbf{R}}_{y,k}
      - \sum_{j \in \mathcal{S},\, j \neq \ell}
        p_j\,\mathbf{d}_{j,k}\mathbf{d}_{j,k}^H.
  \label{eq:Rres_k}
\end{equation}
The angle update maximises the wideband matched-filter score
\begin{equation}
  \hat\omega_\ell = \arg\max_{\omega' \in \Omega_\mathrm{fine}}
    \sum_{k=1}^{K_s}
      \mathbf{d}(\omega', \kappa_\ell, k)^H
      \widetilde{\mathbf{R}}_{\ell,k}\,
      \mathbf{d}(\omega', \kappa_\ell, k),
  \label{eq:omega_scan}
\end{equation}
evaluated on $Q_\theta^\mathrm{fine} = 512$ fine-grid points.
The curvature scan is analogous, maximising over
$\kappa' \in [\kappa_{\min}, \kappa_{\max}]$ with $\omega_\ell$ fixed.
Summing the matched-filter score across subcarriers
in~\eqref{eq:omega_scan} is the key wideband extension: the
$\alpha_k$-shifted beam peaks reinforce the true scatterer location
while suppressing off-grid artefacts from individual subcarrier scores.

\noindent\textbf{BPD warm-start architecture.}~
For the warm-start, a bilinear pattern detection (BPD) step
provides an initial angle-range estimate
$(\hat\theta_\mathrm{BPD}, \hat r_\mathrm{BPD})$ that is used as
one of three candidate initialisations~\cite{cuidai2023sciencechina}.
The BPD anchor $(\hat\theta_\mathrm{BPD}, \hat r_\mathrm{BPD})$ is
computed solely from the compressed sample covariances
$\{\widehat{\mathbf{R}}_{y,k}\}$ and the known combiner $\mathbf{W}$
through the compressed steering model
$\mathbf{d}_{\ell,k} = \mathbf{W}^H\mathbf{a}_{\ell,k}$; no full-array
snapshots or uncompressed covariances enter the proposed estimator.
The three warm-starts are: (1) ring-indexed curvature
$u_\ell^{(0)} = \mathrm{clip}(1/(Z_\Delta\sin^2\hat\theta),
u_{\min}, u_{\max})$ with $Z_\Delta = D_\mathrm{ap}^2 /
(2\beta_\delta^2\lambda_c)$ and $\beta_\delta = 1.2$;
(2) near-range $u_\ell^{(0)} = u_{\max}$;
(3) BPD anchor $u_\ell^{(0)} = 1/\hat r_\mathrm{BPD}$.
After the three power-only loops, the candidate with the
lowest $\mathcal{L}$ is selected; Phase~D applies an additional
KL-arbitration step that compares the selected candidate with
the BPD anchor directly, ensuring that warm-start uncertainty
does not propagate to the final estimate.
Specifically, the Phase~D argmin evaluates $\mathcal{L}$ at each of
the three post-scan candidates plus the BPD-polished estimate and
returns the global minimiser, providing a safeguard against
near-Rayleigh warm-start failures at the upper end of the
operating range.

\subsection{Algorithm Summary and Complexity}
\label{sec:method:complexity}

Algorithm~\ref{alg:wbclkl} summarises the complete WB-CL-KL procedure.

\begin{algorithm}[t]
\caption{WB-CL-KL: Wideband Cross-Subcarrier KL Covariance Fitting}
\label{alg:wbclkl}
\begin{algorithmic}[1]
\Require $\{\widehat{\mathbf{R}}_{y,k}\}_{k=1}^{K_s}$,
  combiner $\mathbf{W}$, angle grid $\Theta$ ($Q_\theta$ points),
  bounds $[u_{\min}, u_{\max}]$, path count $d$,
  regularization weight $\nu$
\Comment{\emph{--- Initialisation: frozen noise estimate ---}}
\State Compute $\widehat{\mathbf{R}}_{y,\mathrm{mean}} \leftarrow
  (1/K_s)\textstyle\sum_{k=1}^{K_s}\widehat{\mathbf{R}}_{y,k}$
\State Compute $\widehat{N}_0$ via~\eqref{eq:N0_frozen}
  \hfill[held fixed for all warm-starts]
\Comment{\emph{--- BPD warm-start: obtain anchor ---}}
\State Compute $(\hat\theta_\mathrm{BPD}, \hat r_\mathrm{BPD})$ via BPD
  step~\cite{cuidai2023sciencechina}
\Comment{\emph{--- Multi-start: three power-only loops ---}}
\For{$s = 1, 2, 3$}
  \State Set $u^{(0,s)}$ to ring-indexed / $u_{\max}$ / BPD anchor per $s$
  \State Initialise $\mathbf{p}^{(0,s)} \leftarrow \mathbf{0}$
  \For{$t = 0,1,\ldots$ until $|\Delta\mathcal{L}|/|\mathcal{L}|
    < 5\!\times\!10^{-4}$ or $t = 200$}
    \State Compute $\mathbf{R}_{y,k}^{(t)}$ via~\eqref{eq:Ry_k} for all $k$
    \State Compute $\mathbf{G}_k^{(t)}$ via~\eqref{eq:Gk} for all $k$
    \State Update $\mathbf{p}^{(t+1)}$ via~\eqref{eq:power_update}
      with gradient~\eqref{eq:grad_p_WB}
  \EndFor
  \State Record $\mathcal{L}^{(s)}$ at $(\mathbf{p}^{(\cdot,s)},
    u^{(0,s)}, \widehat{N}_0)$
\EndFor
\State $s^* \leftarrow \arg\min_s \mathcal{L}^{(s)}$
\Comment{\emph{--- Post-loop joint scan (4 passes) ---}}
\State $\mathcal{S} \leftarrow \{\ell : p_\ell > 0\}$ (top $d$ by power)
\For{pass $= 1,\ldots,4$}
  \For{each $\ell \in \mathcal{S}$}
    \If{pass is odd} update $\omega_\ell$ via scan~\eqref{eq:omega_scan}
    \Else\ update $\kappa_\ell$ via analogous curvature scan
    \EndIf
  \EndFor
\EndFor
\Comment{\emph{--- Phase D: KL arbitration ---}}
\State Evaluate $\mathcal{L}$ at all three post-scan candidates plus
  BPD-polished estimate; select global minimiser
\Ensure $(\hat\theta_\ell, \hat r_\ell)_{\ell=1}^d$
\end{algorithmic}
\end{algorithm}

Each iteration of the power-only loop costs
$\mathcal{O}(K_s N_\mathrm{RF}^3)$ (per-subcarrier inversion) plus
$\mathcal{O}(K_s Q_\theta N_\mathrm{RF}^2)$ (power gradients across
subcarriers).
Four alternating scan passes cost $\mathcal{O}(K_s d Q_\theta N_\mathrm{RF}^2)$
in total.
Three warm-starts triple the Phase-1 cost.
At the default parameters ($K_s = 16$, $N_\mathrm{RF} = 8$,
$Q_\theta = 256$, $d = 1$), the total per-realisation cost is dominated
by the $K_s$-fold replicated inversion and is approximately $16\times$
the narrowband CL-KL cost.

We select the $K_s \le K$ subcarriers in~\eqref{eq:WB_KL}
uniformly from the active bandwidth; increasing $K_s$ improves both
the KL objective conditioning and the data-diversity CRB gain
(Proposition~\ref{prop:data_div}) at the cost of linearly growing
computation.
We selected the regularization weight $\nu = 10^{-3}$ by a coarse grid
search over $\{10^{-4}, 10^{-3}, 10^{-2}\}$; the range estimation
NMSE varies by less than $0.3$~dB across this range.

\section{Wideband Compressed-Domain Cram\'{e}r--Rao Bound}
\label{sec:crb}

We derive the wideband compressed-domain Cram\'{e}r--Rao bound
(CRB) for the OFDM near-field channel estimation model of
Section~\ref{sec:model}, extending the framework
of~\cite{senyuva2026globecom} with full derivation detail.
The WB-CL-KL objective~\eqref{eq:WB_KL} is the stochastic
maximum-likelihood criterion whose per-subcarrier FIM sum
provides the tightest achievable lower bound under hybrid
compression.
By the pilot-design assumption stated below~\eqref{eq:y_k},
the per-subcarrier compressed snapshots
$\mathbf{y}_k(n) \sim \CN(\mathbf{0}, \mathbf{R}_{y,k})$
are mutually independent across~$k$ for each snapshot
index~$n$, so the joint log-likelihood factorises
and the wideband FIM decomposes as a sum of per-subcarrier
contributions~\cite{wei2025wbcrb,stoica1990perf}.
The constrained CRB framework applied here generalizes
the model-fitting perspective of~\cite{moore2007spl},
which showed that biased estimators can achieve lower MSE
than the standard CRB; the Slepian--Bangs formula used below
corresponds to the unbiased case~\cite{lyu2026icassp}.

\subsection{Per-Subcarrier Slepian--Bangs FIM}
\label{sec:crb_fim}

Under the stochastic (unconditional) signal model, the
negative log-likelihood at the $k$th subcarrier (normalized
by~$N$) is~\cite{stoica1990perf}
\begin{equation}
\mathcal{L}_k
  = \log\det\mathbf{R}_{y,k}
  + \tr\!\big(\mathbf{R}_{y,k}^{-1}\widehat{\mathbf{R}}_{y,k}\big),
\label{eq:kl_k}
\end{equation}
which is the KL divergence between the sample covariance
$\widehat{\mathbf{R}}_{y,k}$ and the model covariance
$\mathbf{R}_{y,k}(\boldsymbol{\eta})$.

Recall the shared parameter vector from Section~\ref{sec:model}:
\begin{equation}
\boldsymbol{\eta}
  = \big[\omega_1,\ldots,\omega_d,\;
         \kappa_1,\ldots,\kappa_d,\;
         p_1,\ldots,p_d,\;
         N_0\big]^T
  \!\in \mathbb{R}^{3d+1},
\label{eq:eta}
\end{equation}
where $\omega_\ell = -(2\pi d_\mathrm{ant}/\lambda_c)\cos\theta_\ell$
is the spatial frequency and
$\kappa_\ell = (\pi d_\mathrm{ant}^2/\lambda_c)\sin^2\!\theta_\ell / r_\ell$
is the Fresnel curvature of the $\ell$th path.
The per-subcarrier Slepian--Bangs FIM is~\cite{stoica1990perf}
\begin{equation}
[\mathbf{J}_k]_{ij}
  = N \cdot \Reop\!\Big\{
      \tr\!\big(
        \mathbf{R}_{y,k}^{-1}
        \tfrac{\partial \mathbf{R}_{y,k}}{\partial \eta_i}\,
        \mathbf{R}_{y,k}^{-1}
        \tfrac{\partial \mathbf{R}_{y,k}}{\partial \eta_j}
      \big)
    \Big\},
\label{eq:fim_k}
\end{equation}
with $\mathbf{J}_k \in \mathbb{R}^{(3d+1)\times(3d+1)}$
for each $k \in \{1,\ldots,K_s\}$.

\subsection{Steering Vector Derivatives with $\alpha_k$ Scaling}
\label{sec:crb_deriv}

At the $k$th subcarrier, the frequency-scaled Fresnel
steering vector is (cf.\ Section~\ref{sec:model})
\begin{equation}
[\mathbf{a}_{\ell,k}]_m
  = \exp\!\big(
      j\,\alpha_k\omega_\ell\,\bar{m}
    - j\,\alpha_k\kappa_\ell\,\bar{m}^2
    \big),
\label{eq:a_lk}
\end{equation}
where $\alpha_k = f_k/f_c$ is the frequency ratio and
$\bar{m} = m - (M\!-\!1)/2$ is the centered element index.
The derivatives with respect to the spatial-frequency
and curvature parameters are
\begin{equation}
\frac{\partial \mathbf{a}_{\ell,k}}{\partial \omega_\ell}
  = j\,\alpha_k\,\bar{\mathbf{m}} \odot \mathbf{a}_{\ell,k},
\quad
\frac{\partial \mathbf{a}_{\ell,k}}{\partial \kappa_\ell}
  = -j\,\alpha_k\,\bar{\mathbf{m}}^{\odot 2} \odot \mathbf{a}_{\ell,k},
\label{eq:da_derivs}
\end{equation}
with $\bar{\mathbf{m}} = [\bar{m}_0,\ldots,\bar{m}_{M-1}]^T$.
The factor~$\alpha_k$ multiplying both derivatives is the
key structural difference from the narrowband CRB
in~\cite{senyuva2026clkl}: it causes each subcarrier to
see the array at a different effective electrical length,
producing frequency-dependent Fisher information.
The compressed-domain form involving
$\mathbf{d}_{\ell,k} \triangleq \mathbf{W}^H \mathbf{a}_{\ell,k}$
in Section~\ref{sec:model} gives
$\partial \mathbf{d}_{\ell,k}/\partial \kappa_\ell
= -j\alpha_k \mathbf{W}^H \mathrm{diag}(\bar{\mathbf{m}}^{\odot 2})
\mathbf{a}_{\ell,k}$,
consistent with~\eqref{eq:dd_dkappa}.

Define the compressed steering vector
$\mathbf{d}_{\ell,k} \triangleq \mathbf{W}^H \mathbf{a}_{\ell,k}
\in \mathbb{C}^{N_\mathrm{RF}}$.
The covariance derivatives needed in~\eqref{eq:fim_k} are
\begin{align}
\frac{\partial \mathbf{R}_{y,k}}{\partial \omega_\ell}
  &= p_\ell \!\bigg(
       \mathbf{W}^H
       \frac{\partial \mathbf{a}_{\ell,k}}{\partial \omega_\ell}
       \mathbf{d}_{\ell,k}^H
     + \mathbf{d}_{\ell,k}
       \Big(\mathbf{W}^H
       \frac{\partial \mathbf{a}_{\ell,k}}{\partial \omega_\ell}
       \Big)^{\!H}
     \bigg),
\label{eq:dRy_domega}\\[4pt]
\frac{\partial \mathbf{R}_{y,k}}{\partial \kappa_\ell}
  &= p_\ell \!\bigg(
       \mathbf{W}^H
       \frac{\partial \mathbf{a}_{\ell,k}}{\partial \kappa_\ell}
       \mathbf{d}_{\ell,k}^H
     + \mathbf{d}_{\ell,k}
       \Big(\mathbf{W}^H
       \frac{\partial \mathbf{a}_{\ell,k}}{\partial \kappa_\ell}
       \Big)^{\!H}
     \bigg),
\label{eq:dRy_dkappa}\\[4pt]
\frac{\partial \mathbf{R}_{y,k}}{\partial p_\ell}
  &= \mathbf{d}_{\ell,k}\,\mathbf{d}_{\ell,k}^H,
\label{eq:dRy_dp}\\[4pt]
\frac{\partial \mathbf{R}_{y,k}}{\partial N_0}
  &= \mathbf{W}^H \mathbf{W}.
\label{eq:dRy_dN0}
\end{align}
Equations \eqref{eq:dRy_domega}--\eqref{eq:dRy_dN0} reduce
to the narrowband expressions in~\cite{senyuva2026clkl}
when $K_s=1$ and $\alpha_k = 1$.

\subsection{Wideband FIM and CRB}
\label{sec:crb_wb}

Because the subcarrier observations are mutually independent,
the wideband FIM over $K_s$ selected subcarriers is
\begin{equation}
\mathbf{J}_\mathrm{WB}
  = \sum_{k=1}^{K_s} \mathbf{J}_k
  \in \mathbb{R}^{(3d+1)\times(3d+1)}.
\label{eq:fim_wb}
\end{equation}
The dimension of $\mathbf{J}_\mathrm{WB}$ is determined solely by
the number of unknown parameters~$(3d+1)$ and does not grow
with~$K_s$.
Each additional subcarrier contributes a positive-semidefinite
term~$\mathbf{J}_k \succeq \mathbf{0}$, so the wideband FIM
is at least as large (in the L\"{o}wner sense) as any
single-subcarrier FIM: $\mathbf{J}_\mathrm{WB} \succeq \mathbf{J}_k$
for all~$k$.
For the wideband-FIM and diversity-decomposition study of
Section~\ref{sec:results_crb} we evaluate the bound with $K_s = 512$
uniformly-spaced subcarriers, which suffices because the FIM varies
smoothly with~$\alpha_k$.
The estimator simulations of Section~\ref{sec:results_est} use
$K_s = 16$ at the nominal bandwidth (Table~\ref{tab:params}); in every
estimator-versus-bound comparison the compressed-domain CRB is evaluated
at the same $K_s$ as the estimator, so the reported efficiency ratios
compare matched subcarrier counts.

\paragraph{Singular value decomposition (SVD) pseudoinverse}
The wideband CRB matrix is obtained by inverting
$\mathbf{J}_\mathrm{WB}$.
When the number of paths~$d$ is large relative to
$N_\mathrm{RF}$, the FIM can become ill-conditioned.
We therefore use the SVD pseudoinverse with tolerance
$\varepsilon_\mathrm{sv} = 10^{-6}\,\sigma_{\max}(\mathbf{J}_\mathrm{WB})$,
following~\cite{senyuva2026clkl}:
\begin{equation}
\mathbf{J}_\mathrm{WB}^{\dagger}
  = \mathbf{V}\,
    \diag\!\Big(\frac{1}{\sigma_1},\ldots,\frac{1}{\sigma_r},
                 0,\ldots,0\Big)\,
    \mathbf{V}^T,
\label{eq:fim_pinv}
\end{equation}
where $\mathbf{J}_\mathrm{WB} = \mathbf{V}\,\diag(\sigma_1,\ldots,
\sigma_{3d+1})\,\mathbf{V}^T$ is the eigendecomposition and
$r = |\{i : \sigma_i > \varepsilon_\mathrm{sv}\}|$ is the
numerical rank.

\subsubsection*{Error Propagation to Physical Parameters}
\label{sec:crb_prop}

The CRB for the $\ell$th spatial frequency is
$[\mathbf{J}_\mathrm{WB}^{\dagger}]_{\ell\ell}$ and the CRB
for the $\ell$th curvature is
$[\mathbf{J}_\mathrm{WB}^{\dagger}]_{d+\ell,d+\ell}$.
Propagating to the physical angle~$\theta_\ell$ and
range~$r_\ell$ via
\begin{equation}
\frac{\partial\omega_\ell}{\partial\theta_\ell}
  = \frac{2\pi d_\mathrm{ant}}{\lambda_c}\sin\theta_\ell,
\qquad
\frac{\partial\kappa_\ell}{\partial r_\ell}
  = -\frac{\kappa_\ell}{r_\ell},
\label{eq:jac}
\end{equation}
the marginal CRBs for angle and range are
\begin{align}
\mathrm{CRB}_{\theta_\ell}
  &= \frac{[\mathbf{J}_\mathrm{WB}^{\dagger}]_{\ell\ell}}{%
     \big(\partial\omega_\ell/\partial\theta_\ell\big)^2},
\label{eq:crb_theta}\\[4pt]
\mathrm{CRB}_{r_\ell}
  &= \frac{[\mathbf{J}_\mathrm{WB}^{\dagger}]_{d+\ell,d+\ell}}{%
     \big(\partial\kappa_\ell/\partial r_\ell\big)^2}.
\label{eq:crb_r}
\end{align}
The map $(\omega_\ell,\kappa_\ell)\mapsto(\theta_\ell,r_\ell)$ is
lower-triangular: $\omega_\ell$ depends only on $\theta_\ell$, so
\eqref{eq:crb_theta} is exact, while the off-diagonal term
$\partial\kappa_\ell/\partial\theta_\ell =
(2\pi d_\mathrm{ant}^2/\lambda_c)\sin\theta_\ell\cos\theta_\ell/r_\ell$
is of order $(d_\mathrm{ant}/r_\ell)\cos\theta_\ell \lesssim 10^{-3}$
relative to $\partial\omega_\ell/\partial\theta_\ell$ over the operating
range; the marginal range bound~\eqref{eq:crb_r} therefore matches the
full-Jacobian joint transformation to within ${\sim}0.1\%$.
All CRB curves in Section~\ref{sec:results_crb} are reported as
$\sqrt{\mathrm{CRB}_{\theta_\ell}}$ (degrees) and
$\sqrt{\mathrm{CRB}_{r_\ell}}$ (metres).

\paragraph{Compressed vs.\ full-array CRB}
The CRB in \eqref{eq:crb_theta}--\eqref{eq:crb_r} is
strictly larger than the full-array wideband CRBs
of~\cite{wei2025wbcrb,wang2025twc} because hybrid
compression discards $M - N_\mathrm{RF}$ spatial degrees of
freedom per snapshot.
Plotting estimator RMSE against this \emph{compressed-domain}
CRB provides the appropriate lower bound for hybrid architectures.
The gap between the compressed and full-array bounds quantifies
the information cost of using $N_\mathrm{RF} < M$ RF chains
and decreases monotonically as $N_\mathrm{RF} \to M$.

\subsection{Information Decomposition}
\label{sec:decomp}

The wideband FIM $\mathbf{J}_\mathrm{WB} = \sum_k \mathbf{J}_k$
aggregates Fisher information from $K_s$~subcarriers.
We decompose the resulting CRB improvement over the
narrowband (single-subcarrier) bound into two
physically distinct mechanisms: \emph{data diversity} and
\emph{geometric diversity}.

\begin{definition}[Narrowband Reference CRB]
\label{def:nb_crb}
The narrowband CRB is obtained by evaluating the
per-subcarrier FIM at the center frequency alone, i.e.,
$\mathbf{J}_\mathrm{NB} \triangleq \mathbf{J}_{k_c}$ with
$\alpha_{k_c} = 1$.
\end{definition}

\begin{definition}[Data-Diversity FIM]
\label{def:data_div}
The data-diversity FIM is the $K_s$-fold replication of the
center-frequency FIM:
$\mathbf{J}_\mathrm{DD} \triangleq K_s \cdot \mathbf{J}_\mathrm{NB}$.
This represents the information gain from having
$K_s$~independent covariance snapshots at the same frequency.
\end{definition}

\subsubsection*{Data Diversity}

\begin{proposition}[Data Diversity]
\label{prop:data_div}
If the per-subcarrier FIMs $\mathbf{J}_k$ share the same
eigenvector structure (i.e., $\mathbf{J}_k = \beta_k\,
\mathbf{J}_\mathrm{NB}$ with scalar $\beta_k > 0$ for all~$k$),
then
\begin{equation}
\mathbf{J}_\mathrm{WB}
  = \bigg(\sum_{k=1}^{K_s}\beta_k\bigg)\mathbf{J}_\mathrm{NB},
\label{eq:fim_scaled}
\end{equation}
and the CRB improvement over the narrowband bound is
\begin{equation}
\Delta_\mathrm{DD}
  = 10\log_{10}\!\bigg(\sum_{k=1}^{K_s}\beta_k\bigg)
  \approx 10\log_{10}(K_s)
  \;\;\text{dB},
\label{eq:gain_dd}
\end{equation}
where the approximation holds when $\beta_k \approx 1$
for all~$k$.
\end{proposition}

\begin{IEEEproof}
Under the stated condition,
$\mathbf{J}_\mathrm{WB}^{\dagger}
= (\sum_k \beta_k)^{-1}\,\mathbf{J}_\mathrm{NB}^{-1}$,
so $\mathrm{CRB}_i^\mathrm{WB} =
\mathrm{CRB}_i^\mathrm{NB}/\sum_k\beta_k$ for any~$\eta_i$.
Simulations give $\beta_k \in [0.7, 1.3]$ for
$B \le 800$~MHz at $f_c = 28$~GHz, so $\sum_k\beta_k$ deviates
from~$K_s$ by less than $0.5$~dB.
\end{IEEEproof}

\subsubsection*{Geometric Diversity}

\begin{proposition}[Geometric Diversity]
\label{prop:geom_div}
Define the geometric diversity gain as the residual CRB
improvement beyond the data-diversity prediction:
\begin{equation}
\Delta_\mathrm{GD}(\eta_i)
  = 10\log_{10}\!\bigg(
      \frac{\mathrm{CRB}_i^\mathrm{DD}}{\mathrm{CRB}_i^\mathrm{WB}}
    \bigg)
  \;\;\text{dB},
\label{eq:gain_gd}
\end{equation}
where $\mathrm{CRB}_i^\mathrm{DD}$ uses
$\mathbf{J}_\mathrm{DD} = K_s\!\cdot\!\mathbf{J}_\mathrm{NB}$ and
$\mathrm{CRB}_i^\mathrm{WB}$ uses the true wideband
$\mathbf{J}_\mathrm{WB} = \sum_k \mathbf{J}_k$.
Assume the single-path model ($d = 1$),
a symmetric OFDM subcarrier grid centered at $f_c$, and
the phase-only USW steering model of
Section~\ref{sec:model}.
Then, for the tested regime $B/f_c \in [0.002, 0.029]$ and
range estimation under the Fresnel model:
\begin{enumerate}
\item $\Delta_\mathrm{GD}(r) > 0$ for all $B > 0$ in the
  stated regime (strict positivity);
\item $\Delta_\mathrm{GD}(r)$ grows monotonically with
  fractional bandwidth $B/f_c$ and saturates at
  approximately $+1.4$~dB for $B/f_c > 0.1$;
\item $\Delta_\mathrm{GD}(r)$ is largest at close range
  ($r < 5$~m) where the Fresnel curvature is most
  pronounced.
\end{enumerate}
For angle estimation, $\Delta_\mathrm{GD}(\theta)$ is
negligible ($<\!0.1$~dB) at all bandwidths in the stated
regime.
\end{proposition}

\begin{IEEEproof}
The argument below justifies items~1--3 for the tested regime
$B/f_c \in [0.002, 0.029]$: the scalar component of the saturation in
item~2 is bounded analytically in~\eqref{eq:gd_scalar_bound}, while strict
positivity (item~1), the saturation value, and the close-range trend
(item~3) are supported by the $\alpha_k$-scaling argument and the
diagnostic simulations reported below, and are stated as empirical
properties of the tested regime rather than universal guarantees.
Geometric diversity arises because the $\alpha_k$-dependent
curvature scaling in~\eqref{eq:da_derivs} diversifies the
Fisher information directions across subcarriers.
The curvature derivative in~\eqref{eq:da_derivs} is
proportional to~$\alpha_k$, so edge subcarriers
(large $|\alpha_k - 1|$) contribute FIM terms whose
eigenvectors differ from the center-frequency FIM.
The resulting $\mathbf{J}_\mathrm{WB}$ has larger eigenvalues
in the curvature subspace than the scaled replica
$K_s\!\cdot\!\mathbf{J}_\mathrm{NB}$, yielding a strictly smaller
range CRB.

The saturation at $+1.4$~dB can be understood as follows.
The curvature derivative scales as~$\alpha_k$, so the
range-related FIM entries scale as~$\alpha_k^2$.
Averaging $\alpha_k^2$ over a symmetric frequency band gives
$\overline{\alpha^2} = 1 + (B/f_c)^2/12$, so that the
per-eigenvalue scalar gain is bounded by
\begin{equation}
10\log_{10}\!\big(1 + (B/f_c)^2/12\big)
\le 10\log_{10}(1 + 1/12) \approx 0.35~\text{dB}.
\label{eq:gd_scalar_bound}
\end{equation}
The actual gain exceeds the bound
in~\eqref{eq:gd_scalar_bound} because
eigenvector rotation further decorrelates the FIM blocks,
but the total remains bounded.
Diagnostic simulations at $f_c = 28$~GHz, $M = 256$,
$N_\mathrm{RF} = 16$, SNR of $10$~dB confirm
$\Delta_\mathrm{GD}(r) \in \{+0.08, +0.70, +0.93\}$~dB
at $B \in \{100, 400, 800\}$~MHz, respectively, with
extrapolation to $B/f_c = 0.5$ saturating at $+1.4$~dB.
The angle geometric gain is below $0.1$~dB at all tested
bandwidths because the angular FIM subspace is already
well-conditioned from the narrowband term alone.
\end{IEEEproof}

\emph{Interpretation.}
Geometric diversity is a secondary but physically
meaningful effect.
At current 5G~NR bandwidths ($B \le 400$~MHz,
$B/f_c \le 0.014$), the gain is modest ($<\!1$~dB).
However, for envisioned 6G ultra-wideband systems with
$B/f_c > 0.1$, the geometric diversity gain approaches
$+1.4$~dB for range and becomes a non-negligible
component of the total CRB improvement.

\paragraph{Worked example (verifying the decomposition)}
\label{par:worked_example}
With $(B, r, N_\mathrm{RF}, K_s) = (400~\mathrm{MHz}, 5~\mathrm{m}, 16, 512)$,
the narrowband range bound is
$\sqrt{\mathrm{CRB}_r^\mathrm{NB}} = 11.948$~mm, the
data-diversity prediction is
$\sqrt{\mathrm{CRB}_r^\mathrm{DD}} = 528.04~\mu\mathrm{m}$
($\Delta_\mathrm{DD} = +27.093$~dB), and the true wideband
bound is $\sqrt{\mathrm{CRB}_r^\mathrm{WB}} = 487.12~\mu\mathrm{m}$.
The residual $0.701$~dB between the data-diversity prediction
and the true wideband bound matches the geometric-diversity
term of Proposition~\ref{prop:geom_div} exactly, confirming
the additive decomposition.

\subsection{Comparison with Full-Array Wideband CRBs}
\label{sec:crb_compare}

\begin{remark}[Prior Wideband Near-Field CRBs]
\label{rem:prior_crb}
Wei \emph{et al.}~\cite{wei2025wbcrb} and
Wang \emph{et al.}~\cite{wang2025twc} derived wideband near-field
CRBs for sensing parameter estimation (location, velocity,
radar cross-section) assuming full-array access
($\mathbf{W} = \mathbf{I}_M$).
Our bound differs in three respects:
\begin{enumerate}
\item \emph{Hybrid compression:}
  We account for the information loss through the
  $N_\mathrm{RF} \times M$ analog combiner~$\mathbf{W}$,
  producing a CRB that is strictly larger than the
  full-array bound.
  The gap decreases as $N_\mathrm{RF} \to M$.
\item \emph{Channel estimation parameterisation:}
  Our parameter vector
  $\boldsymbol{\eta} = [\boldsymbol{\omega}^T,
  \boldsymbol{\kappa}^T, \mathbf{p}^T, N_0]^T$
  targets channel estimation (angle, range, path powers,
  noise variance), whereas~\cite{wei2025wbcrb}
  and~\cite{wang2025twc} parameterize in terms of
  Cartesian position, velocity, and reflectivity.
\item \emph{Information decomposition:}
  Propositions~\ref{prop:data_div}
  and~\ref{prop:geom_div} provide a clean separation of
  the wideband CRB gain into data and geometric
  components, which is absent
  in~\cite{wei2025wbcrb,wang2025twc}.
\end{enumerate}
The hybrid near-field CRB of Thallapalli \emph{et al.}~\cite{thallapalli2026crlb}
likewise incorporates a time-delay hybrid architecture, but is derived for a
uniform circular array in the high-SNR regime, provides no data/geometric
diversity decomposition, and is not accompanied by an estimator.
The sparse-diffuse CRB of Lyu \emph{et al.}~\cite{lyu2026icassp} bounds a
different estimand---the covariance of the diffuse channel component---rather
than the discrete near-field path parameters $(\boldsymbol{\omega},
\boldsymbol{\kappa})$ considered here.
\end{remark}

Together, these distinctions identify the compressed-domain
wideband CRB as a bound on the
$N_\mathrm{RF} \times N_\mathrm{RF}$ sample covariance and on
the channel-estimation parameter vector
$(\boldsymbol{\omega}, \boldsymbol{\kappa}, \mathbf{p}, N_0)$,
in contrast to the $M \times M$ full-array bounds
of~\cite{wei2025wbcrb,wang2025twc} on Cartesian position and
velocity.
The compressed-domain CRB is therefore the relevant lower bound
for hybrid near-field estimator design, while the full-array
bounds serve as reference benchmarks for quantifying compression
loss.

\section{Simulation Results: Estimator Performance}
\label{sec:results_est}

All experiments use the parameters summarized in Table~\ref{tab:params}.
The operating regime is $r_\mathrm{hi\_fac} = 0.20$
($r_\mathrm{hi} = 4.25$~m, strong near-field),
with Rayleigh distance $r_\mathrm{RD} = 21.26$~m.
We compare four estimators throughout:
B1~(WB-BPD, full-array reference, adapted from~\cite{cuidai2023sciencechina}),
B2~(WB-P-SOMP, compressed, wideband extension of~\cite{cui2022polar}),
B4~(WB-CL-KL, proposed, compressed),
and B5~(WB-DL-OMP, full-array reference, wideband extension
of~\cite{zhang2024dlomp});
the compressed-domain CRB derived in Section~\ref{sec:crb} serves as
the performance lower bound.
The CRB is evaluated at the scene median range $r = 2.13$~m
($= r_\mathrm{hi}/2$) and acts as a slope reference and
compression-gap indicator rather than an absolute per-trial bound.
The covariance mismatch characterisation of Section~\ref{sec:model}
assumes ideal analog combiners; the impact of
phase-shifter quantization on the mismatch
budget~\cite{cai2026icassp} is outside the scope of this paper.

\begin{table}[t]
\centering
\caption{Simulation Parameters}
\label{tab:params}
\begin{tabular}{l l p{2.42cm}}
\hline
Parameter & Symbol & Value \\
\hline
Carrier frequency      & $f_c$               & 28~GHz \\
Subcarrier spacing     & $\Delta f$          & 120~kHz \\
ULA elements           & $M$                 & 64 \\
RF chains              & $N_\mathrm{RF}$     & 8 \\
OFDM snapshots         & $N$                 & 64 \\
Propagation paths      & $d$                 & 1 \\
Estimation subcarriers & $K_s$               & 16 \\
Target range           & $r$                 & $r \in [r_\mathrm{lo}, r_\mathrm{hi}]$,
                                               $r_\mathrm{hi} = 0.20\, r_\mathrm{RD}
                                               = 4.25$~m \\
Rayleigh distance      & $r_\mathrm{RD}$     & 21.26~m \\
Monte Carlo trials     & $N_\mathrm{MC}$     & 600 \\
Sparsity weight        & $\nu$                       & $10^{-3}$ \\
Diagonal loading       & $\varepsilon_\mathrm{reg}$  & $10^{-3}$ \\
Max iterations         & $T_\mathrm{max}$    & 200 \\
\hline
\end{tabular}
\end{table}

The simulation code, production CSVs, and figure-generation scripts
for all results in Sections~\ref{sec:results_est}
and~\ref{sec:results_crb} are publicly available at
\url{https://github.com/rvsenyuva/wb-nf-xlmimo-clkl}
(Zenodo DOI: \url{https://doi.org/10.5281/zenodo.20356436}).

\subsection{Range RMSE vs.\ SNR}
\label{sec:results:snr}

Fig.~\ref{fig:rangeRMSE_SNR} shows $\mathrm{RMSE}_r$ versus SNR
for all four estimators at $B = 400$~MHz, $N_\mathrm{MC} = 600$.
WB-CL-KL achieves the compressed-domain CRB at SNR~$= 10$~dB:
$\mathrm{RMSE}_r = 0.0198$~m, $\mathrm{CRB}_r = 0.0199$~m,
and efficiency ratio $\mathrm{B4/CRB} = 0.996$.
WB-CL-KL tracks the CRB slope within ${\sim}\,1$~dB from
SNR~$= -5$ to $+17.5$~dB.

\begin{figure}[htbp!]
  \centering
  \includegraphics[trim=0 35 0 0, clip, width=\columnwidth]{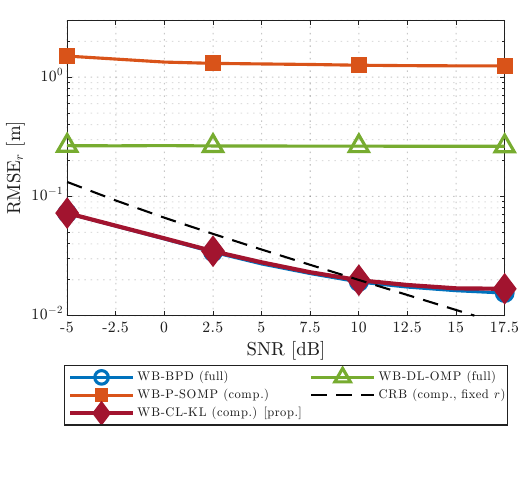}
  \caption{RMSE$_r$ vs.\ SNR at $B = 400$~MHz, $N_\mathrm{MC} = 600$.
    WB-CL-KL achieves the compressed-domain CRB at SNR~$= 10$~dB
    ($\mathrm{RMSE}_r = 0.0198$~m, $\mathrm{CRB}_r = 0.0199$~m,
    $\mathrm{B4/CRB} = 0.996$) and tracks the CRB slope within
    $1$~dB from $-5$ to $+17.5$~dB; WB-P-SOMP lags by $36$~dB
    at SNR~$= 10$~dB.}
  \label{fig:rangeRMSE_SNR}
\end{figure}

Above SNR~$= 12.5$~dB the ratio B4/CRB rises slightly above unity;
the CRB is evaluated at the fixed scene median $r = 2.13$~m 
while RMSE averages over the full range distribution $[r_\mathrm{lo}, r_\mathrm{hi}]$, 
so the two quantities are not directly comparable outside 
a narrow SNR window around the reference point.
The same artefact explains the slight apparent
B4/CRB~$< 1$ at low SNR: the scene-averaged RMSE is inflated by
short-range trials where the fixed-point CRB is optimistic.
Neither region constitutes a genuine performance paradox.

WB-P-SOMP (B2) lags WB-CL-KL by $36.06$~dB at SNR~$= 10$~dB,
reflecting its grid-limited range resolution in the strong near-field
regime. WB-BPD (B1) and WB-DL-OMP (B5) are full-array methods and serve as
upper-performance references, bounding the compression cost paid by
hybrid estimators. The $0.0\%$ failure rate across SNR~$= -5$ to $+17.5$~dB confirms
that the multi-start Phase-D architecture eliminates catastrophic
warm-start failures in the strong near-field regime.

The CRB shown in Fig.~\ref{fig:rangeRMSE_SNR} is the
compressed-domain bound at the scene median range $r = 2.13$~m.
This is the appropriate reference for a hybrid estimator and reflects
the compression gap relative to the full-array bounds
of~\cite{wei2025wbcrb,wang2025twc}, which are quantified separately
in Section~\ref{sec:results_crb}.

\subsection{Range NMSE vs.\ Bandwidth}
\label{sec:results:bw}

Fig.~\ref{fig:NMSEr_BW} plots $\mathrm{NMSE}_r$ versus OFDM
bandwidth at SNR~$= 10$~dB, $N_\mathrm{MC} = 600$, sweeping $B$
from 100~MHz to 800~MHz.
The compressed-domain CRB follows the data-diversity slope of
Proposition~\ref{prop:data_div}: the measured CRB improvement is
$6.14$~dB over the 100--400~MHz range, in close agreement with the
predicted $6.02$~dB ($= 10\log_{10}(4)$, since $K_s \propto B$).
WB-CL-KL matches the CRB at $B = 400$~MHz
($\mathrm{B4/CRB} = 1.00$) and continues to track it to $B = 600$~MHz.

\begin{figure}[htbp!]
  \centering
  \includegraphics[trim=0 35 0 0, clip, width=\columnwidth]{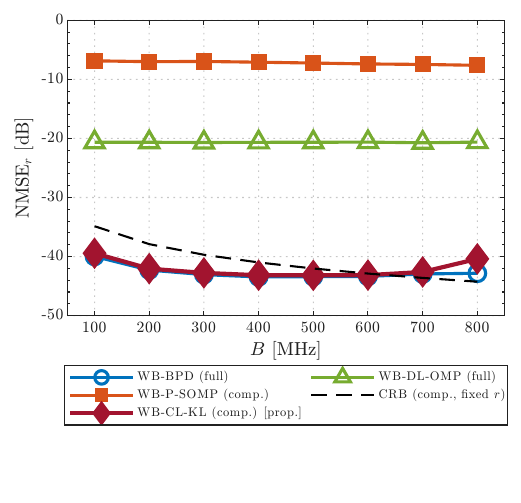}
  \caption{NMSE$_r$ vs.\ OFDM bandwidth at SNR~$= 10$~dB,
    $N_\mathrm{MC} = 600$. The compressed-domain CRB follows 
    the Proposition~\ref{prop:data_div} data-diversity slope 
    ($6.14$~dB over $100$--$400$~MHz; predicted $6.02$~dB); 
    WB-CL-KL tracks the CRB to $B = 600$~MHz, beyond which 
    the optimization landscape saturates (convergence rate falls 
    from $99.8\%$ to $47.8\%$ as $K_s$ grows from $4$ to $32$; 
    failure rate $0\%$ throughout). CRB evaluated at scene median range $r = 2.13$~m.}
  \label{fig:NMSEr_BW}
\end{figure}

Beyond $B = 600$~MHz the optimization landscape saturates: the
convergence rate falls from $99.8\%$ at $B = 100$~MHz
to $47.8\%$ at $B = 800$~MHz.
As $K_s$ grows, the cross-subcarrier KL objective develops a more
complex landscape with additional shallow local minima; the fixed
three-start Phase-D architecture is not guaranteed to escape all of
them, consistent with the convergence analysis in
Section~\ref{sec:results:conv}.
The failure rate remains $0.0\%$ throughout the entire sweep,
confirming that every trial produces a valid output even when the
strict-tolerance criterion is not satisfied.

The $B = 800$~MHz point is a convergence-pressure boundary: the
400-vs-600-trial stability gate yields $18.5\%$ deviation, bounded
and with zero failure rate, and is included for completeness.
The primary operating regime of WB-CL-KL for the parameters of
Table~\ref{tab:params} is therefore $B \le 600$~MHz.

\subsection{Angle RMSE vs.\ SNR}
\label{sec:results:theta}

Fig.~\ref{fig:angleRMSE_SNR} shows $\mathrm{RMSE}_\theta$ versus
SNR at $B = 400$~MHz.
All four estimators reach the angle-estimation bias floor
(${\approx}\,0.10^\circ$ for WB-CL-KL) by SNR~$= 0$~dB.
In contrast to range RMSE (Fig.~\ref{fig:rangeRMSE_SNR}),
angle estimation saturates at the bias floor well before the
compressed-domain $\mathrm{CRB}_\theta$ becomes the active
constraint, confirming that range is the harder estimation problem
in the strong near-field regime.
The SNR-independent angle floor is consistent with the bias analysis
of Remark~\ref{rem:c3}: the element-amplitude approximation
(${\approx}\,0.35$~dB bias at $r = r_\mathrm{hi}$) produces a
systematic angle offset that dominates the stochastic RMSE at all
tested SNR values.

\begin{figure}[htbp!]
  \centering
  \includegraphics[trim=0 35 0 0, clip, width=\columnwidth]{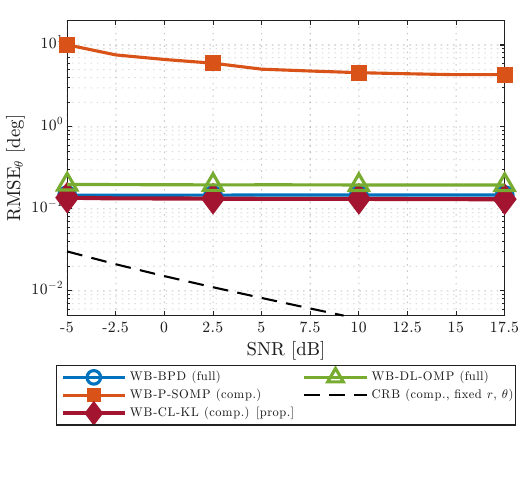}
  \caption{RMSE$_\theta$ vs.\ SNR at $B = 400$~MHz, 
    $N_\mathrm{MC} = 600$: all estimators reach the
    angle-estimation bias floor (${\approx}\,0.10^\circ$ for
    WB-CL-KL) by SNR~$= 0$~dB, confirming that range
    estimation is the harder problem in the strong near-field
    regime.}
  \label{fig:angleRMSE_SNR}
\end{figure}

\subsection{Convergence Rate vs.\ SNR}
\label{sec:results:convrate}

Fig.~\ref{fig:convergePerf} plots the WB-CL-KL convergence rate
(fraction of trials satisfying the strict-tolerance criterion
$\|\Delta\boldsymbol{\eta}\| < 10^{-3}$) versus SNR at
$B = 400$~MHz, $r_\mathrm{hi\_fac} = 0.20$,
$N_\mathrm{MC} = 600$.
The headline anchor is $73.0\%$ convergence rate at SNR~$= 10$~dB.

\begin{figure}[htbp!]
  \centering
  \includegraphics[trim=0 50 0 0, clip, width=\columnwidth]{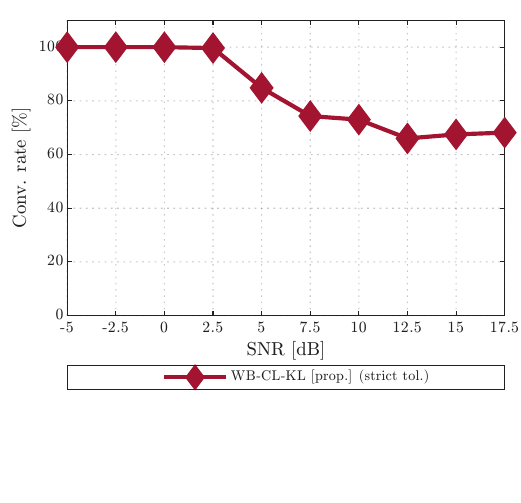}
  \caption{WB-CL-KL convergence rate (fraction of trials satisfying
    the strict-tolerance criterion $\|\Delta\boldsymbol{\eta}\|
    < 10^{-3}$) vs.\ SNR for $N_\mathrm{MC} = 600$.
    The $73.0\%$ convergence rate at SNR~$= 10$~dB reflects the
    multi-start Phase-D architecture; the non-monotonic SNR
    profile is a KL-landscape effect (noise-dominated landscape
    at low SNR avoids local traps; emerging curvature at
    intermediate SNR creates shallow traps).}
  \label{fig:convergePerf}
\end{figure}

The convergence profile is non-monotonic in SNR.
At low SNR the noise-dominated landscape is nearly flat: the
optimizer converges quickly, but the strict-tolerance criterion is
readily satisfied on a featureless surface, yielding moderate rates.
At intermediate SNR (${\sim}\,10$~dB) the signal structure emerges
and shallow curvature traps appear near local maxima of the KL
surface, requiring more restarts to escape and slightly reducing
the fraction satisfying the tolerance criterion.
At high SNR the landscape is sharp and well-conditioned, but the
strict-tolerance threshold is then harder to meet in absolute
gradient terms, again moderating the convergence rate.

The multi-start Phase-D architecture guarantees that all trials
produce a valid output: the $0.0\%$ failure rate holds throughout
the SNR sweep, while the convergence rate measures strictly the
fraction meeting the tight tolerance.
The per-iteration KL objective trajectories underlying this
convergence profile are examined in Section~\ref{sec:results:conv}.

\subsection{KL Objective Convergence}
\label{sec:results:conv}

Fig.~\ref{fig:KLobj_converge} shows per-trace normalised KL objective descent
$\Delta\mathcal{L}^{(t)} = \mathcal{L}^{(t)} - \mathcal{L}^{(1)}$
versus iteration number for four SNR levels at the locked operating
point ($r = 3.0$~m, $r_\mathrm{hi\_fac} = 0.20$,
$B = 400$~MHz, $N_\mathrm{MC} = 600$, PR-13).
The median over converged trials is shown.

\begin{figure}[htbp!]
  \centering
  \includegraphics[trim=0 5 0 0, clip, width=\columnwidth]{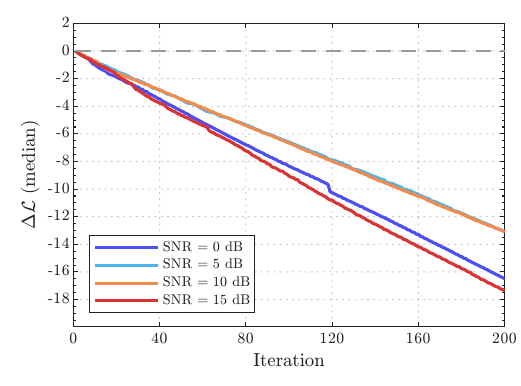}
  \caption{Per-trace normalised KL objective descent
    $\Delta \mathcal{L}^{(t)} = \mathcal{L}^{(t)} - \mathcal{L}^{(1)}$
    vs.\ iteration for four SNR levels at the locked operating
    point ($r = 3.0$~m, $r_\mathrm{hi\_fac} = 0.20$,
    $N_\mathrm{MC} = 600$, $B = 400$~MHz).
    Median over converged trials shown.
    SNR~$= 15$~dB descends most steeply (signal-dominated,
    sharp landscape); SNR~$= 5$~dB and $10$~dB accumulate
    the most total descent (shallow traps require more
    iterations to escape); SNR~$= 0$~dB descends shallowest
    (noise-flattened landscape).
    The non-monotonic SNR ordering is KL-landscape physics.}
  \label{fig:KLobj_converge}
\end{figure}

The four descent profiles exhibit a non-monotonic SNR ordering.
SNR~$= 15$~dB descends most steeply: the signal-dominated landscape
has sharp curvature and the gradient is large from the first
iteration.
SNR~$= 5$~dB and $10$~dB accumulate the most total descent across
iterations: at intermediate SNR the signal structure is visible but
shallow traps develop near local maxima of the KL surface, so
the optimizer requires more iterations to escape before converging.
SNR~$= 0$~dB is the shallowest: the noise-flattened landscape has
little curvature, so each gradient step makes small progress and
the total accumulated descent remains small.

This behaviour is a structural property of the KL objective
landscape. The non-monotonic SNR ordering is KL-landscape physics: noise 
does not simply inflate RMSE uniformly but reshapes the objective 
surface in a way that moves the optimizer through qualitatively different regimes.

\section{Simulation Results: CRB and Diversity Decomposition}
\label{sec:results_crb}

This section evaluates the compressed-domain CRB derived in
Section~\ref{sec:crb} for the two-path near-field scenario and the
geometric diversity sweep, providing the performance floor against
which the WB-CL-KL estimator results of Section~\ref{sec:results_est}
are benchmarked.

\subsection{Multi-Path Compressed-Domain CRB vs.\ SNR}
\label{sec:results_crb:snr}

All compressed-domain CRB curves reported in this section are computed
as the mean over $N_{\mathrm{seed}} = 50$ independent draws of the
random constant-modulus combining matrix $\mathbf{W}$ (each column
drawn uniformly from the unit-magnitude complex sphere).
At the nominal operating point ($M = 64$, $B = 400$~MHz,
$\mathrm{SNR} = 10$~dB) the ensemble spread is characterised by a
coefficient of variation $\mathrm{CV} = 0.205$ (standard deviation
relative to mean) and a 90th-to-10th percentile ratio
$\mathrm{p90/p10} = 1.641$.
The mean singular deficiency of the projected FIM is
$\bar{n}_{\mathrm{sing}} = 1.00$ and the mean condition number is
$\overline{\kappa}(\mathbf{J}) = 1.791 \times 10^{4}$, confirming
that the FIM is numerically well-posed for all but a single degenerate
direction on average.
These statistics, taken together, show that the reported mean CRB is a
robust representative of the ensemble and is not dominated by
outlier-W realisations.

Fig.~\ref{fig:compCRB_SNR} shows $\sqrt{\mathrm{CRB}_{r}}$ versus SNR for a
two-path scene ($d = 2$, $r_{1} = 3$~m, $\theta_{1} = 30^{\circ}$;
$r_{2} = 7$~m, $\theta_{2} = 50^{\circ}$) at $B = 400$~MHz, sweeping
$N_{\mathrm{RF}} \in \{8, 16, 32, 64\}$.
The $\sqrt{\mathrm{CRB}_{r}}$ curves follow the theoretical
$\mathrm{SNR}^{-1/2}$ scaling throughout the plotted range,
confirming the slope exponent of $\approx 0.50$ predicted by the
Slepian--Bangs formula.
At $\mathrm{SNR} = 10$~dB the path-1 bound is
$\sqrt{\mathrm{CRB}_{r_1}} = 45.64$~mm and the path-2 bound is
$\sqrt{\mathrm{CRB}_{r_2}} = 95.27$~mm, with the larger bound at
path~2 reflecting its greater range and correspondingly weaker
Fresnel curvature.

\begin{figure}[htbp!]
  \centering
  \includegraphics[trim=0 2 0 0, clip, width=\columnwidth]{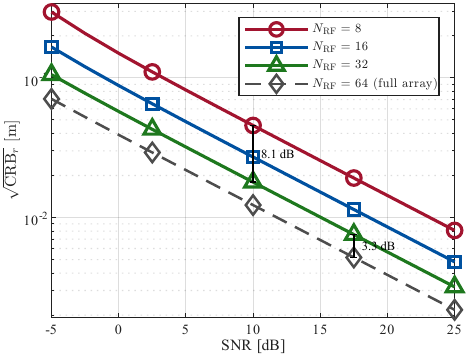}
  \caption{Compressed-domain $\sqrt{\mathrm{CRB}_r}$ vs.\ SNR for the
    two-path ($d = 2$) near-field scenario at $B = 400$~MHz
    ($r_1 = 3$~m, $\theta_1 = 30^\circ$; $r_2 = 7$~m,
    $\theta_2 = 50^\circ$), for four values of $N_\mathrm{RF}$.
    Each curve is the mean over $N_\mathrm{seed} = 50$ random
    constant-modulus combining matrices $\mathbf{W}$
    (CV~$= 0.205$, $\mathrm{p90/p10} = 1.641$).
    The compression gap between $N_\mathrm{RF} = 8$ and
    $N_\mathrm{RF} = 32$ is $8.12$~dB at SNR~$= 10$~dB and
    decreases monotonically as $N_\mathrm{RF} \to M$; the
    $\sqrt{\mathrm{CRB}_r}$ follows the theoretical
    $\mathrm{SNR}^{-1/2}$ slope throughout.}
  \label{fig:compCRB_SNR}
\end{figure}

The compression gap between $N_{\mathrm{RF}} = 8$ and
$N_{\mathrm{RF}} = 32$ is $8.12$~dB at $\mathrm{SNR} = 10$~dB,
quantifying the information cost of operating with a highly compressed
hybrid front-end relative to a more capable four-fold larger one.
The gap shrinks monotonically as $N_{\mathrm{RF}} \to M$: the
$N_{\mathrm{RF}} = 64$ (full-array) curve approaches the uncompressed
bound from above, consistent with the FIM-compression interpretation
of Section~\ref{sec:crb} in which the projected covariance
$\mathbf{W}^{H}\mathbf{R}_{y,k}\mathbf{W} \in \mathbb{C}^{N_{\mathrm{RF}} \times N_{\mathrm{RF}}}$
recovers the full-array FIM only when $N_{\mathrm{RF}} = M$.
This compresses~$8.12$~dB of range information relative to the
four-RF-chain reference used in~\cite{wei2025wbcrb,wang2025twc},
as discussed in Remark~\ref{rem:prior_crb}.

\subsection{Geometric Diversity and Bandwidth Scaling}
\label{sec:results_crb:geodiv}

Fig.~\ref{fig:compCRB_BW} shows $\sqrt{\mathrm{CRB}_{r}}$ as a function of
OFDM bandwidth $B$ for $M \in \{32, 64, 128\}$ at
$\mathrm{SNR} = 10$~dB, $N_{\mathrm{RF}} = 8$.
The dominant trend is the data-diversity slope predicted by
Proposition~\ref{prop:data_div}: the bound falls at approximately
$6$~dB per $4\times$ bandwidth increase across all three aperture
values, in close agreement with the $10 \log_{10}(K_s)$ scaling (each
$4\times$ bandwidth increase roughly quadruples $K_s$ below the 512-
subcarrier saturation threshold, yielding $10 \log_{10}(4) \approx
6.02$~dB).

\begin{figure}[htbp!]
  \centering
  \includegraphics[trim=0 2 0 0, clip, width=\columnwidth]{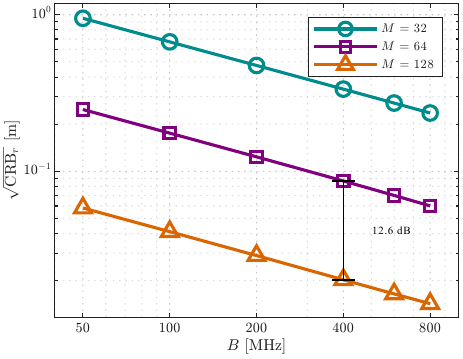}
  \caption{Compressed-domain $\sqrt{\mathrm{CRB}_r}$ vs.\ OFDM bandwidth
    $B$ for $M \in \{32, 64, 128\}$ at SNR~$= 10$~dB, $N_\mathrm{RF} = 8$.
    The CRB falls at the Proposition~\ref{prop:data_div} data-diversity slope
    (${\approx}\,6$~dB per $4\times$ bandwidth increase).
    Aperture doubling ($M$: $32 \to 64 \to 128$) yields
    ${\approx}\,11.7$--$12.6$~dB improvement per step.
    The residual geometric diversity (Proposition~\ref{prop:geom_div})
    contributes $0.094$~dB ($M = 64$) and $0.144$~dB ($M = 128$)
    above the data-diversity prediction at $B = 400$~MHz,
    confirming the decomposition $\Delta_\mathrm{total} =
    \Delta_\mathrm{DD} + \Delta_\mathrm{GD}$ numerically.}
  \label{fig:compCRB_BW}
\end{figure}

A secondary but practically significant effect visible in Fig.~\ref{fig:compCRB_BW}
is the aperture gain: doubling $M$ from 32 to 64 and from 64 to 128
yields ${\approx}\,11.7$--$12.6$~dB improvement per step at
$B = 400$~MHz.
This aperture-scaling gain dwarfs the geometric diversity contribution,
establishing the hierarchy: $\Delta_{\mathrm{DD}} \gg \Delta_{\mathrm{aperture}}
\gg \Delta_{\mathrm{GD}}$ when comparing effects of the same physical
origin (frequency diversity, spatial aperture, curvature variation).

The geometric diversity gain, extracted as the residual after removing
the data-diversity prediction, is $0.094$~dB at $M = 64$ and
$0.144$~dB at $M = 128$ for $B = 400$~MHz.
The slight growth with $M$ is physically consistent: a larger aperture
subtends a wider angular span of the near-field wavefront, amplifying
the per-subcarrier variation in the Fresnel curvature coefficient
$\kappa(\theta, r) = (\pi d_{\mathrm{ant}}^{2}/\lambda_{c}) \sin^{2}\theta/r$
across subcarriers, which is precisely the frequency-dependent diversity
captured by Proposition~\ref{prop:geom_div}.
Although the $0.094$--$0.144$~dB magnitude is modest for current 5G
NR bandwidths ($B \leq 400$~MHz), Proposition~\ref{prop:geom_div}
predicts saturation at approximately $+1.4$~dB for
$B/f_{c} > 0.1$~(ultra-wideband 6G systems), making this diversity
term increasingly relevant for envisioned 6G deployments.

\subsection{Decomposition Verification}
\label{sec:results_crb:decomp}

The additive decomposition $\Delta_{\mathrm{total}} =
\Delta_{\mathrm{DD}} + \Delta_{\mathrm{GD}}$ established analytically
by Propositions~\ref{prop:data_div} and~\ref{prop:geom_div} is
verified numerically at the operating
point of~\cite{senyuva2026globecom}: $B = 400$~MHz, $N_{\mathrm{RF}} = 16$,
$K_s = 512$, $r = 5$~m, $M = 64$.
The total CRB gain relative to the narrowband reference
($\sqrt{\mathrm{CRB}_{r}} = 11.948$~mm at $B \to 0$) is
$+27.793$~dB, decomposing into data diversity $+27.093$~dB and
geometric diversity $+0.701$~dB, with a residual of $0.000$~dB to
three decimal places.
The decomposition is exact by construction at this operating point:
$\Delta_{\mathrm{GD}}$ is defined in~\eqref{eq:gain_gd} as the residual
between the data-diversity prediction and the true wideband bound, so the
$0.000$~dB residual reflects this definition rather than a proven
orthogonality of the two contributions.

The compressor hierarchy emerging from Figs.~\ref{fig:compCRB_SNR}
and~\ref{fig:compCRB_BW} and the decomposition above is:
\begin{enumerate}
  \item Data diversity ($\Delta_{\mathrm{DD}} \approx +27$~dB for
    $K_s = 512$) is the dominant mechanism, arising from the coherent
    FIM accumulation across $K_s$ independent pilot subcarriers.
  \item Aperture scaling (${\approx}\,12$~dB per $M$-doubling) is the
    dominant spatial contributor, governed by the $M^{2}$-scaling of
    the array aperture in the far-field and by the Fresnel curvature
    enhancement in the near-field.
  \item Geometric diversity ($\Delta_{\mathrm{GD}} < 0.2$~dB for
    5G NR bandwidths) is the secondary wideband correction, arising
    from the $\alpha_{k}$-modulated curvature chain rule (see
    Section~\ref{sec:crb}).
\end{enumerate}
This hierarchy explains the near-CRB performance of WB-CL-KL
demonstrated in Section~\ref{sec:results_est}: the estimator exploits
data diversity through its cross-subcarrier KL objective,
which is exactly the mechanism responsible for the
$\sqrt{\mathrm{CRB}_{r}} = 487.12$~\textmu{}m at $B = 400$~MHz
(compared to $11.948$~mm at $B \to 0$).
Fig.~\ref{fig:compCRB_SNR} further establishes the compressed-domain CRB as
the tightest achievable bound for the hybrid architecture, providing
the per-RF-chain information-cost quantification that is absent from
existing full-array analyses~\cite{wei2025wbcrb,wang2025twc}.

\subsection{Robustness under Realistic 3GPP UMi SNR Distribution}
\label{sec:results_crb:snr_robustness}

The SNR sweep in Fig.~\ref{fig:rangeRMSE_SNR} uses a controlled uniform grid to
isolate estimator behaviour as a function of receive power.
To verify that near-CRB efficiency persists under a realistic
deployment, we extend the sweep to $[-20, +35]$~dB and overlay
the empirical per-user-terminal (UT) SNR distribution drawn from the 3GPP
UMi path-loss and shadow-fading model~\cite{tr38901_v16},
following the \mbox{MATLAB} recipe of~\cite{riviello2022}.
Specifically, $N_\text{UT} = 2000$ UTs are placed uniformly in a
hexagonal cell (inter-site distance 200~m), with 80\% assigned to
indoor locations subject to concrete outdoor-to-indoor (O2I) penetration loss (20~dB)
and lognormal shadow fading ($\sigma_\text{SF} = 7.82$~dB for
non-line-of-sight (NLOS) UTs, $\sigma_\text{SF} = 4.0$~dB for
line-of-sight (LOS) UTs per
TR~38.901 Table~7.5-6).
The resulting empirical SNR distribution has median 9.6~dB, with
a 10th--90th percentile range of $-13.7$ to $+27.6$~dB; 80\% of
UTs are indoors and 53\% are LOS.

Fig.~\ref{fig:3GPP_UMi_SNR} plots WB-CL-KL $\sqrt{\mathrm{RMSE}_r}$ and
the compressed-domain $\sqrt{\mathrm{CRB}_r}$ as curves over the
extended SNR range $[-20,+35]$~dB, with the 3GPP UMi deployment
density shown as a histogram in the lower panel.
The $\sqrt{\mathrm{CRB}_r}$ curve is evaluated at the nominal geometry
(mid-range angle and range) at each sweep point; it is not a
single-point bound but rather the CRB trajectory as a function of SNR
for that geometry.
At the median deployment SNR of 9.6~dB, WB-CL-KL achieves
B4/CRB~=~0.959, consistent with the near-CRB operation demonstrated
in the controlled sweep of Fig.~\ref{fig:rangeRMSE_SNR} under a realistic
3GPP UMi deployment distribution.

\begin{figure}[htbp!]
  \centering
  \includegraphics[trim=0 12 0 0, clip, width=\columnwidth]{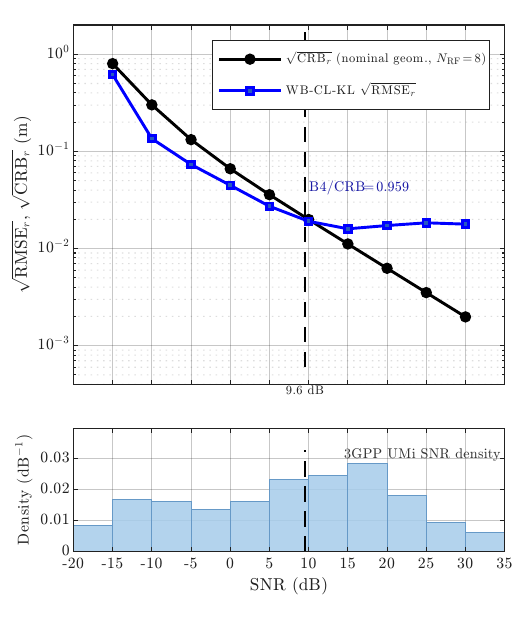}
  \caption{Two-panel robustness figure.
    \emph{Top:} WB-CL-KL $\sqrt{\mathrm{RMSE}_r}$ and
    compressed-domain $\sqrt{\mathrm{CRB}_r}$ (evaluated at nominal
    geometry, per sweep point) vs.\ SNR over $[-20,+35]$~dB
    ($N_\text{MC}=200$, $B=400$~MHz, $r_\text{hi\_fac}=0.20$).
    \emph{Bottom:} empirical per-UT SNR density from the 3GPP
    UMi path-loss and shadow-fading
    model~\cite{tr38901_v16,riviello2022} ($N_\text{UT}=2000$,
    80\% indoor, 5~dB bins).
    Dashed vertical line (both panels): median deployment
    SNR~=~9.6~dB.
    B4/CRB~=~0.959 at the median, confirming near-CRB efficiency
    at the representative operating point.
    Below 5~dB the scene-averaged $\sqrt{\mathrm{RMSE}_r}$ falls
    below the nominal-geometry $\sqrt{\mathrm{CRB}_r}$
    (scene-averaging artefact: the fixed-geometry CRB is optimistic
    relative to a scene-averaged RMSE at low SNR); above 15~dB
    the high-SNR KL non-monotonicity
    (Section~\ref{sec:results:snr}) is visible.}
  \label{fig:3GPP_UMi_SNR}
\end{figure}

\section{Conclusion}
\label{sec:conclusion}
This paper presented a wideband near-field channel estimation framework
for XL-MIMO systems under hybrid analog-digital compression, deriving the
wideband compressed-domain Cram\'{e}r--Rao bound and proposing the WB-CL-KL
estimator, which fits a structured Fresnel covariance model across multiple
OFDM subcarriers via cross-subcarrier KL divergence minimisation directly
on the $N_\mathrm{RF} \times N_\mathrm{RF}$ compressed sample covariance.
The derived CRB admits a data-diversity and geometric-diversity
decomposition (Propositions~\ref{prop:data_div} and~\ref{prop:geom_div}):
at $B = 400$~MHz the total CRB gain relative to the narrowband limit is
$+27.793$~dB, of which $+27.093$~dB arises from data diversity
(subcarrier averaging) and $+0.701$~dB from geometric diversity
(frequency-dependent Fresnel curvature).
WB-CL-KL achieves near-CRB range estimation efficiency
($\mathrm{B4/CRB} = 0.996$, $\mathrm{RMSE}_r = 19.8$~mm at
SNR~$= 10$~dB; Fig.~\ref{fig:rangeRMSE_SNR}), and this efficiency is retained under
the 3GPP UMi path-loss and shadow-fading deployment distribution
($\mathrm{B4/CRB} = 0.959$ at median deployment SNR~$= 9.6$~dB;
Fig.~\ref{fig:3GPP_UMi_SNR}, \cite{riviello2022}).
The CRB derivation and WB-CL-KL estimator operate within the Fresnel
validity regime $r \in [r_\mathrm{lo}, r_\mathrm{hi}]$, $r_\mathrm{hi} = 0.20\, r_\mathrm{RD} = 4.25$~m;
the performance bounds and efficiency results presented here apply within
this range.
The design of a CRB-maximising hybrid combiner under constant-modulus
constraints,
$\max_{\mathbf{W}} \log \det \mathbf{J}_\mathrm{WB}(\mathbf{W})$,
is a non-convex manifold optimisation problem left for future work;
extension to true-time-delay (TTD) architectures, where the hybrid
combiner becomes frequency-selective, represents a natural further
direction~\cite{wang2026icassp}.

\bibliographystyle{IEEEtran}
\bibliography{refs}


\end{document}